\documentclass[preprint,12pt]{article}

\usepackage[cmex10]{amsmath}
\usepackage{amssymb}
\usepackage{amsthm}
\usepackage{amsfonts}

\usepackage{algorithmic}
\usepackage{algorithm}

\usepackage{url}

\newtheorem{prop}{Proposition}
\newtheorem{lem}{Lemma}

\theoremstyle{definition}
\newtheorem{mdef}{Definition}
\newtheorem{ex}{Example}
\theoremstyle{remark}

\newcommand{\ocap}{\mbox{\footnotesize \textcircled{$\scriptstyle{\cap}$}}}
\newcommand{\ocup}{\mbox{\footnotesize \textcircled{$\scriptstyle{\cup}$}}}

\usepackage{hyperref}

\newcommand\pare[1]{\left( #1\right)}
\newcommand\braces[1]{\left\lbrace #1\right\rbrace} 

\usepackage{color}

\definecolor{darkgreen}{RGB}{0,155,0}

\usepackage{array}
\usepackage{graphicx}
\usepackage{authblk}

\begin{document}

\title{Complementary Lipschitz continuity results for the distribution of intersections or unions of independent random sets in finite discrete spaces}

\author[]{John~Klein}

\affil[]{\small Univ. Lille, CNRS, Centrale Lille, UMR 9189 - CRIStAL - Centre de Recherche en Informatique Signal et Automatique de Lille, F-59000 Lille}

\date{}

\maketitle

\begin{abstract}

We prove that intersections and unions of independent random sets in finite spaces achieve a form of Lipschitz continuity. More precisely, given the distribution of a random set $\Xi$, the function mapping any random set distribution to the distribution of its intersection (under independence assumption) with $\Xi$ is Lipschitz continuous with unit Lipschitz constant if the space of random set distributions is endowed with a metric defined as the $L_k$ norm distance between inclusion functionals also known as commonalities. Moreover, the function mapping any random set distribution to the distribution of its union (under independence assumption) with $\Xi$ is Lipschitz continuous with unit Lipschitz constant if the space of random set distributions is endowed with a metric defined as the $L_k$ norm distance between hitting functionals also known as plausibilities.

Using the epistemic random set interpretation of belief functions, we also discuss the ability of these distances to yield conflict measures. All the proofs in this paper are derived in the framework of Dempster-Shafer belief functions. Let alone the discussion on conflict measures, it is straightforward to transcribe the proofs into the general (non necessarily epistemic) random set terminology. 
 
\end{abstract}


\section{Introduction}

When one is interested in a point-valued random variable but has access to set-valued (imprecise) observations of this latter, Dempster \cite{Dem67} proposed to use a probabilistic model relying on multi-valued mappings. This model was further developed by Shafer \cite{Sha76} in a self-contained framework known today as Dempster-Shafer theory, evidence theory or belief function theory. In this framework, the uncertainty on the value of the random variable is equivalently\footnote{These functions are in bijective correspondence \cite{Sha76}.} captured (among others) by set functions known as the belief, plausibility and commonality functions. These functions evaluate respectively how likely it is that the imprecise observations imply / are consistent / are implied by some event. In this paper, we focus on the case where these events are disjunctions of elements of a finite and discrete space.

Besides, belief functions are also known to be formally equivalent to random sets \cite{Ngu78} and are interpretable as epistemic ones \cite{COUSO2014}. 
A random set is a random element whose realizations are set-valued. 
The probability masses governing the random set can also be uniquely characterized by set functions that are capacities \cite{matheron1975random} (non additive measures). 
Some of these set functions are:
\begin{itemize}
	\item the containment functional that captures the probabilities that a given set contains the random set,
	\item the hitting functional (or capacity functional) that captures the probabilities that the random set intersects a given set,
	\item the inclusion functional that captures the probabilities that a given set is included in the random set.
\end{itemize}
The above functions are the respective random set terminology for the belief, plausibility and commonality functions.

In the deterministic setting, one can observe a form of consistency between operations like union or intersection with set-distances in the sense that, for some of these distances, if one intersects (resp. unites) the same set with two other ones, say $X_1$ and $X_2$, then the obtained intersections (resp. unions) are at least as close as $X_1$ and $X_2$ were before. 

In this paper, we investigate if this observed consistency propagates to some extent to the random setting. We prove that if one intersects (resp. unites) the same random set with two other ones, say $\Xi_1$ and $\Xi_2$, then the obtained random intersections (resp. unions) have distributions that are at least as close as the distributions of $\Xi_1$ and $\Xi_2$ were before. These results are dependent on the chosen metric for random set distributions. We examine metrics that consist in $L_k$ norm based distances between either hitting or inclusion functionals. In addition, the results can be rephrased as Lipschitz continuity (with a unit Lipschitz constant) for functions that map any random set distribution to the distribution of its intersection (resp. union) with a given fixed random set.

 
In the belief function framework, the closest related works are those of Loudahi et al. \cite{Loudahi2014,Lou16}. The authors proved that the consistency under study holds between some belief function distances and combination operators that yield the distribution of the intersection or union of independent random sets. 
The results that we introduce in this paper involves distances that are computationally more tractable than those introduced in \cite{Loudahi2014,Lou16} but also rely on independence assumptions. In the random set literature, many results regarding unions of i.i.d. random sets and random set metrics are available \cite{molchanov2005theory} but they do not address Lipschitz continuity. Also, we do not require the examined random sets to be identically distributed.

Furthermore, building upon recent work from Pichon and Jousselme \cite{pichon2019several}, we also investigate if our results can be instrumental to span new degrees of conflict. We prove that the consistency of a distance with the conjunctive rule makes the corresponding conflict degree compliant with at least one requirement discussed in Destercke and Burger \cite{des2013}. However, we also show that several distances relying on an $L_k$ norm are not appropriate to yield a degree of conflict as suggested in \cite{pichon2019several} when $k$ is finite. These last developments focus on information fusion aspects of the belief function framework only and do not generalize to non-epistemic random sets.

This article is organized as follows: section \ref{sec:2} gives necessary background on the theory of belief functions and random sets. 
Section \ref{sec:dist} is an overview of distances between random set distributions and the sought Lipschitz continuity property is stated. Section \ref{sec:consistency_with_conj} contains the main results of the paper, i.e. Lipschitz continuity for the distribution of intersection and union of independent random sets. Finally, in section \ref{sec:conflict}, we make use of the aforementioned results to investigate if the examined distances can yield relevant degrees of conflict in the belief function framework. All the proofs of the newly introduced results are given in the appendices.

Most of the paper is written using belief function terminology and usual notations in this framework but it can be transcribed to the random set framework by merely switching the set function names as explained in the first paragraphs of this introduction. Special care was paid to allow easy readability for readers familiar with any of these two frameworks.

\section{Belief functions and random sets}
\label{sec:2}
In this section, some mathematical notations for baseline belief function and random set concepts are given. The reader is expected to be familiar with one of these frameworks. More material on belief functions basics is found for instance in~\cite{Sha76,Cuz14b} and on random sets in \cite{molchanov2005theory,nguyenintroduction,couso2014random}.

Belief functions can be applied in the context of uncountable spaces \cite{Ngu78,Sme05,nguyenintroduction,Den14} but a majority of results were derived in the finite case and we also make this assumption in this article. 

\subsection{Baseline definitions}

A random set in a finite and discrete space $\Omega=\braces{\omega_t}_{t=1}^n$ is a random element whose realizations are subsets of $\Omega$. 
When one is interested in a point valued variable but has access to set valued (imprecise) observations, one can try to infer the distribution of an epistemic random set \cite{miranda2005random}. Belief functions are in line with this epistemic interpretation. When one is interested in a set valued variable and has access to corresponding samples, one can try to infer the distribution of an ontic random set \cite{matheron1975random}. 



Both types of uncertainty lead to formally equivalent objects although these objects need occasionally to be processed and understood in different ways \cite{COUSO2014}. 
In the finite (and consequently countable) setting, the distribution of a random set $\Xi_i$ is a set function called \textbf{mass function} and is denoted by $m_i$. The power set $2^\Omega$ is the set of all subsets of $\Omega$ and it is the domain of mass functions. For any $A\in 2^\Omega$, the \textbf{cardinality} of this set is denoted by $|A|$ and we thus have $|\Omega|=n$. The cardinality of $2^\Omega$ is denoted by $N=2^n$. Mass functions have $\left[0,1\right]$ as co-domain and they sum to one: $\sum_{A\in 2^\Omega}m_i\left(A\right)=1$. A \textbf{focal element} of a mass function $m_i$ is a set $A\subseteq \Omega$ such that $m_i(A)>0$. 
A mass function having only one focal element $A$ is called a \textbf{categorical mass function} and is denoted by $m_A$. 
A \textbf{simple mass function} is the convex combination of $m_\Omega$ with some categorical mass function $m_A$.

Several alternative set functions are commonly used as equivalent characterizations of $\Xi_i$. 
The \textbf{belief}, \textbf{plausibility} and \textbf{commonality} functions of a set $A$ are defined as 
\begin{eqnarray}
bel_i(A) &=& \sum_{E \subseteq A,E \neq \emptyset} m_i(E),\\
pl_i(A) &=& \sum_{E \cap A \neq \emptyset} m_i(E),\\
q_i(A) &=& \sum_{E \supseteq A} m_i(E)
\end{eqnarray}
and respectively represent how much likely it is that $A$ contains / intersects / is included in the underlying random set. In the random set literature, these set functions are respectively referred to as the  \textbf{containment}, \textbf{hitting} and \textbf{inclusion functionals}. When the empty set has a positive mass, another representation is provided by \textbf{implicability} functions $b_i$. These functions are closely related to belief and plausibility functions through the following relations: $\forall A \in 2^\Omega$,
\begin{eqnarray}
b_i(A) &=& bel_i(A) + m_i(\emptyset), \\
b_i(A) &=& 1 - pl_i\pare{A^c}\label{eq:pltob}.
\end{eqnarray}

Another useful concept is the \textbf{negation} (or complement) $\overline{m}_i$ of a mass function $m_i$ introduced by Dubois and Prade \cite{Dub86b}. The function $\overline{m}_i$ is such that $\forall A \subseteq \Omega$, $\overline{m}_i(A)=m_i(A^c)$ with $A^c=\Omega \setminus A$. The authors also provide a result that will be instrumental in the proof of proposition \ref{prop:q_dist_disj}. This result reads
\begin{equation}\label{eq:q_and_b}
	\overline{b}_i \left( A^c \right) = q_i \left( A \right) , \forall A \subseteq \Omega,
\end{equation}
where $\overline{b}_i$ denotes the implicability function in correspondence with $\overline{m}_i$.


\subsection{Intersection, union and information fusion}

Information fusion in the framework of belief functions is performed using an operator mapping an arbitrary large set of input mass functions to a single output mass function which summarizes all information contained in the input ones. 
On top of this minimal requirement, the operator must also follow a certain policy in the way that the information encoded in the input mass functions is processed to build the output one. There are two canonical and dual such  policies: conjunctivity and disjunctivity. 

Suppose $\sqsubseteq$ denotes an informational partial order \cite{yager1986entailment,Dub86b,Den07} for mass functions in the sense that one writes $m_1 \sqsubseteq m_2$ if $m_1$ contains at least as much (epistemic) information as $m_2$. 
Following \cite{dubois2016basic}, a fusion operator is conjunctive if its output is more informative than any input. The conjunctive rule operator \cite{Sme94} denoted  by $\ocap$ is defined as follows 

\begin{align}
m_1 \ocap m_2 \left( E \right)  = \sum_{\substack{A,B \subseteq \Omega \\ A \cap B = E}} m_1 \left( A \right) m_2 \left( B \right) , \forall E\subseteq \Omega.
\end{align}
The conjunctive rule is associative and commutative and the generalization of the above expression to more than two input mass functions is immediate. This rule is the unnormalized version of Dempster's rule \cite{Dem67} and on the random set side, it can be understood as the distribution of the intersection of two independent\footnote{There are several notions of independence for random sets \cite[chapter 2]{couso2014random}. In this paper, we only consider the usual probabilistic notion, i.e. joint distributions factorizing as the product of their marginals. } random sets. Obviously, this operator is a generalization of the set intersection as we have $m_A \ocap m_B = m_{A\cap B}$ for any two subsets $A$ and $B$ of $\Omega$. For the sake of equation concision we adopt the following notation $m_{1\cap 2} = m_1 \ocap m_2 $. This combination is very simple to compute when dealing with commonality functions:
\begin{equation}
	q_{1\cap2} \left( A \right) = q_1 \left( A \right)  q_2 \left( A \right), \forall A \subseteq \Omega.
\end{equation}

It can be easily proved that the conjunctive rule is conjunctive. Consider the informational partial order based on commonalities $\sqsubseteq_q$ which is defined as
\begin{align}
m_1 \sqsubseteq_q m_2 \Leftrightarrow q_1 \left( E \right) \leq q_2 \left( E \right) , \forall E \subseteq \Omega.
\end{align}
Since $q_{1\cap 2}$ is the elementwise multiplication of $q_1$ and $q_2$, we obtain $m_{1\cap 2} \sqsubseteq_q m_1$ and $m_{1\cap 2} \sqsubseteq_q m_2$.

When $m_2=m_E$, i.e. $m_2$ is categorical, the result of the conjunctive combination between $m_1$ and $m_2$ is referred to as the conditioning of $m_{1}$ given $E$ because this operation is a generalization of probabilistic conditioning\footnote{If focal elements of $m_1$ are singletons, i.e. the random set is point valued, then Dempster's conditioning coincides with Bayes rule: $m_{1|E}\left( A \right) = \frac{m_1 \left( A \right) }{m_1 \left( E \right) } $.} \cite{Sha76}. The mass function $m_1 \ocap m_E$ is also denoted by $m_{1|E}$. The following property of the implicability function w.r.t. conditioning will be instrumental in some proofs:
\begin{lem}\label{lem:implic}
For any mass function $m_1$, any categorical mass function $m_E$ and any subset $A\subseteq \Omega$, we have
\begin{equation}\label{eq:cond_b}
	b_{1|E} \left( A \right) = b_1 \left( \left( E\setminus A \right)^c  \right) = b_1 \left( E^c \cup A \right) .
\end{equation}
\end{lem}
To the best of our knowledge, this property is not reported in the belief function literature, we thus provide a  proof in \ref{ap:proof_of_lemma_lem:implic}.

As for disjunction, the output is required to be less informative than any input and it is thus considered as an extremely conservative fusion policy. The disjunctive rule operator \cite{Sme94} denoted  by $\ocup$ is defined as follows 

\begin{align}
m_1 \ocup m_2 \left( E \right)  = \sum_{\substack{A,B \subseteq \Omega \\ A \cup B =E}} m_1 \left( A \right) m_2 \left( B \right) , \forall E\subseteq \Omega.
\end{align}
The disjunctive rule is also associative and commutative and it is a generalization of set union as we have $m_A \ocup m_B = m_{A\cup B}$ for any two subsets $A$ and $B$ of $\Omega$. 
We denote by $m_{1\cup 2} $ the result of the following combination: $ m_1 \ocup m_2 $. 
On the random set side, $m_{1\cup 2} $ is understood as the distribution of the union of two independent random sets. 
The disjunctive combination is very simple to compute when dealing with implicability functions:
\begin{equation}
	b_{1\cup2} \left( A \right) = b_1 \left( A \right)  b_2 \left( A \right), \forall A \subseteq \Omega.
\end{equation}
The disjunctivity of this rule can be proved using the partial order based on implicabilities $\sqsubseteq_b$. This latter reads
\begin{equation}
	m_1 \sqsubseteq_b m_2 \Leftrightarrow b_1 \left( E \right) \geq b_2 \left( E \right) , \forall E \subseteq \Omega.
\end{equation}
 Since $b_{1\cup 2}$ is the elementwise multiplication of $b_1$ and $b_2$, we obtain the desired conclusion.

The disjunctive rule is related to the conjunctive rule by the following De Morgan relation \cite{Dub86b}: for any mass functions $m_1$ and $m_2$
\begin{equation}
	\overline{m_1 \ocap m_2} = \overline{m}_1 \ocup \overline{m}_2. \label{eq:de_morgan}
\end{equation}

 In \ref{ap:alpha_junction}, we mention a more general family of combination rules for belief functions which encompasses the conjunctive and disjunctive rules. These rules are known as $\alpha$-junctions \cite{Sme97}. Since $\alpha$-junctions have a more limited impact in the belief function literature than the conjunctive and disjunctive rules, we chose not to mention them in the main body of this article. However, the results that we prove in the next sections do propagate to $\alpha$-junctions. See \ref{ap:alpha_junction} for the corresponding proofs.

\subsection{The space of mass functions}
Mass functions can be viewed as vectors belonging to the vector space $\mathbb{R}^N$ with categorical mass functions as base vectors. Since mass functions sum to one, the set of mass functions is the simplex $\mathcal{M}$ in that vector space whose vertices are the base vectors $\braces{m_A}_{A\subseteq \Omega}$. This simplex is also called \textbf{mass space}~\cite{Cuz08} and has finite Lebesgue measure but contains uncountably many mass functions. 

Embedding mass functions in a vector space is particularly useful when computing either $m_{1\cap 2}$ or $m_{1\cup 2}$ because they can be obtained as the dot product of some matrix with one of the input mass functions (seen as a column vector) \cite{Sme02}. 
Each such matrix is in one-to-one correspondence with the other mass function. The vector form of any set function will be denoted using bold characters, for instance, the vector form of a mass function $m_i$ is denoted by $\mathbf{m}_i$.

Let $\mathbf{S}_{1}$ denote the \textbf{specialization} matrix \cite{Dub86b} in bijective correspondence with $m_1$. 
Each entry of $\mathbf{S}_1$ is given by $S_1(A,B)=m_{1|B}\left(A\right)$. 
From a geometric point of view~\cite{Cuz10}, each column of $\mathbf{S}_1$ corresponds to the vertex of a polytope $\mathcal{P}_1$, called the \textbf{conditional subspace} of $m_1$. Any mass function $m \in \mathcal{P}_1$ is the result of the combination of $m_1$ with another mass function using $\ocap$. Most importantly, for any mass functions $m_1$ and $m_2$, we have
\begin{equation}
	\mathbf{m}_{1\cap2} = \mathbf{S}_1 \cdot \mathbf{m_2}.
\end{equation}

Let $\mathbf{G}_{1}$ denote the \textbf{generalization matrix} in  bijective correspondence with $m_1$. 
Each entry of $\mathbf{G}_1$ is given by $G_1(A,B)=m_{1\cup B}\left(A\right)$. For any mass functions $m_1$ and $m_2$, we have
\begin{equation}
	\mathbf{m}_{1\cup2} = \mathbf{G}_1 \cdot \mathbf{m_2}.
\end{equation}

There are also transfer matrices allowing to turn mass functions in commonality or implicability functions using a right-handed dot product. They are presented in more details in \ref{ap:alpha_junction}.


\section{From mass function metrics to Lipschitz continuity}\label{sec:dist}
In this section, we will first recall the definitions of some existing distances between mass functions. We focus on (full) metrics and do not discuss dissimilarities \cite{TESSEM1993,Z-D} which have fewer baseline properties as compared to metrics. The Lipschitz continuity property that we seek will then be stated and its desirability will be justified by analyzing set-distances.

\subsection{Vector-based distances}\label{subsec:review}
A distance, or metric, provides a positive real value assessing the discrepancies between two elements. Let us first give a general definition of such an application when the compared vectors are mass functions:
\begin{mdef}
Given a domain $\Omega$ and its related mass space $\mathcal{M}$, a mapping $d:\mathcal{M} \times \mathcal{M} \longrightarrow [0,a]$ with $a\in\mathbb{R}^+$ is a \textbf{distance} between two mass functions $m_1$ and $m_2$ defined on $\Omega$ if the following properties hold:
\begin{itemize}
\item Symmetry : {$d(m_1,m_2) = d(m_2,m_1)$,}
\item Definiteness : {$d(m_1,m_2) = 0 \Leftrightarrow m_1=m_2$,}
\item Triangle inequality : {$d(m_1,m_2)\leq d(m_1,m_3)+d(m_3,m_2)$. }
\end{itemize} \label{Def1}
\end{mdef}
If the mapping fails to possess some of the above properties, then it degrades into unnormalized distance, dissimilarity or pseudo-distance. Only full metrics are able to provide a positive finite value that matches the intuitive notion of \textit{gap}\footnote{This term was used by Frechet~\cite{Fre1906} in his early works on metric spaces, \textit{i.e.} spaces endowed with a distance.} between elements of a given space. \\
If $a\neq +\infty$, then the distance is bounded and if in addition $a=1$, the distance is \textbf{normalized}. Provided that a mass function distance $d$ is bounded, this distance can be normalized by dividing it with $\rho = \max_{A,B\in 2^\Omega}d\pare{m_A,m_B}$ which is the \textbf{diameter} of $\mathcal{M}$ \cite{Lou16}. 

The most popular metric in the belief function literature is Jousselme distance \cite{Jou01}. It is based on an inner product relying on a similarity matrix. This distance is given by:
\begin{equation}
d_{J}(m_1 , m_2) = \sqrt{ \frac{1}{2}\:(\mathbf{ m_1 - m_2})^T\cdot \mathbf{D}\cdot (\mathbf{ m_1 - m_2})}, \label{dJ}
\end{equation}
where $\mathbf{m}_i$ denotes the column vector version of mass function $m_i$ and $\bf D$ is the Jaccard similarity matrix \cite{jac1901} between focal elements. Its components are:
\begin{equation}
D(A,B)=\begin{cases}
1 & \mbox{ if } A=B=\emptyset \\
\frac{|A\cap B|}{|A\cup B|} & \mbox{ otherwise}
\end{cases} .
\end{equation}
 Thanks to the matrix $\bf D$, Jousselme distance takes into account the dependencies between the base vectors of $\mathcal{M}$. Consequently, the poset structure of $ \left(2^\Omega , \sqsubseteq \right)$  has an impact on distance values, allowing a better match with the user's expectations.

Many other mass function distances are defined similarly by substituting matrix $\bf D$ with another matrix evaluating the similarity between base vectors in different ways~\cite{Dia06,Cuz08}. Experimental material in~\cite{Jou12} shows that these distances are highly correlated to $d_J$.


Observe that the aforementioned distances are the $L_2$ norm of the difference of two vectors which are obtained by applying the same linear mapping to each mass function under comparison. We can thus build other distances by resorting to other norms. In particular, when the linear mapping maps a mass function to its corresponding plausibility, commonality or implicability function\footnote{See \cite{Sme02} for the definition of this linear mappings.}, we obtain distances that will be instrumental in the sequel of this paper. The formal definition of these distances follows.
\begin{mdef}\label{def:vec_dist}
  For some family $f \in \left\{ q, bel, pl, b \right\}$ of set functions in bijective correspondence with mass functions, an $L_k$ norm based $\mathbf{f}$\textbf{-distance} $d_{f,k}$ is the following mapping:
\begin{equation}
\begin{array}{cccc}
 d_{f,k} : & \mathcal{M} \times \mathcal{M}  &\rightarrow &\left[0,1\right],\nonumber\\
                & \pare{m_1, m_2}                 &\rightarrow &\frac{1}{\rho} \left\|\mathbf{f}_{1} - \mathbf{f}_{2}\right\|_k.
\end{array}
\end{equation}
$\mathbf{f}_{i}$ is the vector representation of the set function $f_i$ (in correspondence with $m_i$) and $\rho$ is a normalization factor given by 
$$\rho=\max_{A,B\in 2^\Omega}\left\| \mathbf{f}_{A} - \mathbf{f}_{B} \right\|_k.$$
\end{mdef}
For any vector $\mathbf{f}\in \mathbb{R}^N$, its $L^k$ norm are given by:
\begin{eqnarray}
 \left\| \mathbf{f}\right\|_{k} &=& \left( \sum_{A \subseteq \Omega}|f \left( A \right) |^k \right)^{\frac{1}{k}}.\label{eq:Lk}
\end{eqnarray}

Given relation \eqref{eq:pltob}, we see that $d_{pl,k} = d_{b,k}$ for any $k$. Consequently, we do not further mention distances between implicability functions in the sequel of this article. We end this subsection with a small result giving closed form expressions for constant $\rho$ for distances between plausibilities and commonalities.

\begin{lem}
For $L_k$ norm based distance between commonality or  plausibility functions, we have
\begin{equation}
	\rho = \begin{cases} \left( N-1 \right)^{1/k} & \text{ if } k<\infty\\ 1 & \text{ if } k=\infty \end{cases}.
\end{equation}
\end{lem}
\begin{proof} (sketch)
Given proposition 2 in \cite{KLEIN201615}, for any of the distances evoked in the lemma, we have 
\begin{equation}
	\max_{A,B\in 2^\Omega} d\left( m_A,m_B\right) = d\left( m_\Omega,m_\emptyset\right).
\end{equation}
Finally, for $f \in \left\{ q,pl \right\}$, we always have $|f_\Omega \left( A \right) - f_\emptyset \left( A \right)| =1$ if $A\neq \emptyset$ and $|f_\Omega \left( \emptyset \right) - f_\emptyset \left( \emptyset \right)| =0$.
\end{proof}


\subsection{Matrix-based distances} 
\label{sub:matrix_based_evidential_distances}

Since specialization and generalization matrices are also in bijective correspondence with mass functions, we can use the same recipe as in definition \ref{def:vec_dist} to build new mass function distances. The only difference is that mass functions are mapped to matrices and one must thus resort to matrix norms instead of vector norms. Such distances were first introduced in \cite{Loudahi2014,Lou16}. A subset of these distances are defined as follows:

\begin{mdef}\label{def:mat_dist}
  The $L_k$ norm based \textbf{specialization distance} $d_{\text{spe},k}$ is the following mapping:
\begin{equation}
\begin{array}{cccc}
 d_{\text{spe},k} : & \mathcal{M} \times \mathcal{M}  &\rightarrow &\left[0,1\right],\nonumber\\
                & \pare{m_1, m_2}                 &\rightarrow &\frac{1}{\rho} \left\|\mathbf{S}_{1} - \mathbf{S}_{2}\right\|_k.
\end{array}
\end{equation}
$\mathbf{S}_{i}$ is the specialization matrix in correspondence with $m_i$ and $\rho$ is a normalization factor given by 
$$\rho= \begin{cases} \left( 2 \left(    N-1 \right)\right)^{1/k} & \text{ if } k<\infty\\ 1 & \text{ if } k=\infty \end{cases}$$
\end{mdef}
For any matrix $\mathbf{F}\in \mathbb{R}^N$, its $L^k$ norm is given by:
\begin{eqnarray}
 \left\| \mathbf{F}\right\|_{k} &=& \pare{\sum_{A,B \subseteq \Omega}|F \left( A,B \right) |^k}^{\frac{1}{k}}.\label{eq:Lk}
\end{eqnarray}
It was proved in \cite{Lou16} that if we use generalization matrices in the same way as the above definition, we obtain a distance that coincides with the specialization distance.

Other matrix norms were investigated in \cite{Loudahi2014,Lou16}, i.e. operator norms. These norms lead to mass function distances that have fewer desirable properties\footnote{These distances are not consistent with informational partial orders that generalize set inclusion. See \cite{KLEIN201615} for a definition of the consistency of mass function distances with partial orders.} as compared to $L_k$ matrix norm based ones. They are thus not mentioned in this article.


\subsection{Union, intersection and set-distances} 
\label{sub:a_definition_of_consistency_between_distances_and_combinations_rules}

There are two main types of metrics between sets \cite{lesot2008similarity}: those accounting for how many elements are shared by the subsets and those that also account for the number of elements that they do not share. Examples of each category are the following:
\begin{itemize}
	
	\item the Jaccard distance 
	\begin{equation}
d_{\text{jac}}\left( A,B \right)  = \begin{cases} 0 & \text{ if } A=B=\emptyset \\ 1 - D \left( A,B \right) =\frac{|A \Delta B|}{|A \cup B|} & \text{ otherwise }  \end{cases},\nonumber\end{equation}
\item the (normalized) Hamming set distance $d_{\text{ham}} \left( A,B\right) = \frac{|A \Delta B|}{n} $,
\end{itemize}
 where $\Delta$ denotes the set symmetric difference.

The Jaccard distance belongs to the first type of metric while the Hamming distance belongs to the second one. The following example illustrates their main difference.

\begin{ex}
In this example, we replace subsets by their binary representations, i.e. $A=0011$ means that $n=4$ and the elements of $A$ are the third and fourth elements of $\Omega$. We have
\begin{align}
d_{\text{jac}}\left( 0011,0110 \right) = \frac{2}{3} & \text{ and } d_{\text{ham}} \left( 0011,0110 \right) = \frac{2}{4},\\
\text{while } d_{\text{jac}}\left( 00011,00110 \right) = \frac{2}{3} & \text{ and } d_{\text{ham}} \left( 00011,00110 \right) = \frac{2}{5}.
\end{align}
The hamming distance decreases as a fifth element is contained in $\Omega$ and thus this distance depends on elements that the subset do not share.

\end{ex}

We can wonder how these distances interact with set operations like intersection and union. Actually, the nature of these interactions are highly dependent on how one wishes to perform information fusion using either intersections or unions:
\begin{itemize}
	\item Suppose subsets represent a collection of candidate contents that a classifier must assign to an input image. If we want to evaluate if two images have similar contents, we can use a set distance between their imprecise tags. Suppose $\text{image}_a$ is tagged as $\left\{ \text{cat or dog} \right\}$ and $\text{image}_b$ is tagged as $\left\{ \text{dog or bike} \right\}$. If we learn from a second classifier that both images contain pictures of an animal, then we deduce that the image contents are more likely to be closer after inserting this information. 
	More formally, if one intersects both $A$ and $B$ with a third party subset $C$, the result of these intersections cannot be more distant than $A$ and $B$ were initially, which reads

	$$(\text{a})\;\;\; d \left( A\cap C, B \cap C \right) \leq d \left( A, B\right). $$

	\item Suppose subsets are lists of attributes of some streaming video service users. Suppose that it is known that $\text{user}_a$ likes action movies, is a male and lives in the US. Suppose $\text{user}_b$ likes comedies, is a male and lives in the UK. Suppose we learn that both of them also like science fiction movies, then we deduce that these users have closer profiles than previously thought. More formally, if one unites both $A$ and $B$ with a third party subset $C$, the result of these unions cannot be more distant than $A$ and $B$ were initially, which reads

	$$(\text{b})\;\;\;d \left( A\cup C, B \cup C \right) \leq d \left( A, B\right) .$$

\end{itemize}
Observe that in both of these examples, one adopts a conjunctive information fusion policy in the sense that aggregation results are more informative than each input. The conjunctive/disjunctive nature of an operator (like intersection or union) depends on the type of underlying uncertainty. What these examples are meant to highlight is that, should you intend to combine the informative content of subsets using either intersections or unions, then a set distance should comply to either (a) or (b) in order to translate in the numerical distance values that informative contents are more similar after fusion.

It can be proved that the Hamming set distance verifies (a) and (b) while the Jaccard distance verifies (b) only, c.f. \ref{ap:set_dist} for more details. From the applicative contexts of the above examples, we see that the desirability of property (a) or (b) depends on the information fusion operator. Outside the scope of information fusion, we may not require any of these properties for a set metric. 
When one intends to perform information fusion with random sets, it makes sense to wonder if some mass function distances can generalize these properties with respect to information fusion operators defined for them.

\subsection{Union, intersection and random set distribution distances} 
\label{sub:union_intersection_and_random_set_distribution_distances}


One way to generalize the properties $(\text{a})$ or $(\text{b})$ to random sets is stated by the following property:

\begin{mdef}\label{def:weak_consist}
Let $\ast$ be a combination operator and $d$ a mass function distance. $d$ is said to be \textbf{consistent} with respect to $\ast$ if any of the following conditions is verified: 
\begin{itemize}
	\item[(i)] for any mass functions $m_1, m_2$ and $m$ on $\Omega$:
\begin{equation}
 d\left( m_1 \ast m , m_2 \ast m \right) \leq d\left( m_1,m_2\right). \label{eq:consit}
\end{equation}
	\item[(ii)] for any mass function $m$, the mapping $F_m : \mathcal{M} \longrightarrow \mathcal{M}$ of the form
\begin{equation} F_m \left( m_0 \right) = m_0 \ast m \end{equation}
is Lipschitz continuous with $1$ as Lipschitz constant.
\end{itemize}

\end{mdef}
The equivalence between the two conditions follows from the very definition of Lipschitz continuity with $1$ as Lipschitz constant. Indeed, for mapping $F_m$ to qualify as such, it means that we have

$$ d \left( F_m \left( m_1 \right) , F_m \left( m_2 \right)  \right) \leq d \left( m_1,m_2 \right)  $$
for any pair of mass functions $\left( m_1,m_2 \right) 	$. 
In the remainder of this article, we will refer to this property either as consistency property between an operator and distance or, by small abuse of language, as \textbf{1-Lipschitz continuity} of an operator w.r.t. to a distance.

Under this property, repeated combinations with a given mass function $m$ cannot pull away any pair of mass functions. Such mappings are also called non expansive or short maps. Lipschitz continuity is stronger than uniform continuity. In particular, it implies a form of regularity for the corresponding combination mechanism in the sense that the norm of its gradient is bounded by $1$ meaning that the combined mass function does not change very fast or wiggle in the vicinity of functions $m_0$ or $m$.  

From an informative content standpoint, this property also has an impact. Suppose a mass function $m$ is separable\footnote{Shafer \cite{Sha76} introduced this terminology for decompositions w.r.t. Dempster's rule but we understand it in a more general perspective here by considering decompositions w.r.t. some arbitrary rule $\ast$.}, i.e. the combination under rule $\ast$ of elementary pieces of information embodied by simple mass functions yields function $m$. Using a consistent distance w.r.t. $\ast$, mass functions are all the closer as their decompositions involve identical elementary components.


Proving that a fusion operator achieves Lipschitz continuity is not trivial because $\mathcal{M}$ is not finite but instead a compact subset of an uncountable space. In \cite{Loudahi2014}, Loudahi et al. established the consistency of the $L_1$ and $L_{\infty}$ based specialization distances w.r.t. the conjunctive and disjunctive rules. Numerical experiments also show that the $L_2$ based specialization distance is not consistent w.r.t. the conjunctive or disjunctive rule in the sense of definition \ref{def:weak_consist}. The experiments also show that Jousselme distance is not consistent w.r.t. the conjunctive rule. Its consistency with the disjunctive rule remains undetermined.

Observe that we compare random sets only in distribution here. One could think of other generalizations of properties $(\text{a})$ or $(\text{b})$ in probability or in an almost sure fashion.

\section{Lipschitz continuity results}
\label{sec:consistency_with_conj}

 In this section, we provide new Lipschitz continuity results of $L_k$ norm based distances between commonality or plausibility functions with the conjunctive and disjunctive rules.

 \subsection{Main results on Lipschitz continuity for the conjunctive rule} 
 \label{sub:main_results_on_consistent_distance_w_r_t_}
 

\begin{prop}\label{prop:q_dist_conj}
 For $1 \leq k \leq \infty$, $\ocap$ is 1-Lipschitz continuous w.r.t. the $L_k$ norm based $q$-distance $d_{q,k}$.
\end{prop}

See \ref{ap:q_dist_conj} for proof.

\begin{prop}\label{prop:pl_dist_conj}
 $\ocap$ is 1-Lipschitz continuous w.r.t. the $L_\infty$ norm based $pl$-distance $d_{pl,\infty}$.
\end{prop}

See \ref{ap:pl_dist_conj} for proof.

As compared to previous Lipschitz continuity results \cite{Loudahi2014,Lou16}, specialization distances have a greater time complexity as compared to commonality ones. Indeed, although the construction of specialization distances can be sped up \cite{Lou14b}, the time complexity for the specialization distance is quadratic in $N$. More precisely, the time complexity to build a specialization matrix is $O \left(   N^{\frac{\log \left( 3 \right) }{\log \left( 2 \right) }}\right) \approx  O \left(   N^{1.58}\right)$. However, computing the norm of such a matrix has time complexity $O \left( N^2 \right) $. The time complexity to compute a commonality function \cite{kennes2013computational} is $ O \left(    N \log \left( N \right) \right)$ while that of computing the norm of commonality is $O \left( N \right) $. Given relations \eqref{eq:pltob} and \eqref{eq:q_and_b}, the time complexity to compute a plausibility function is identical to that of commonality ones. Moreover, the memory complexity is obviously reduced as well.

\subsection{Lipschitz continuity for the conjunctive rule: numerical experiments and counter-examples}\label{subsec:consist_exp}
This subsection contains experiments illustrating the (in)consistency of several mass function distances with respect to the conjunctive rule. We generate randomly \cite{bur13} $1e4$ triplets of simple mass functions and check how many times inequality \eqref{eq:consit} is verified (when $*=\ocap$) for several distances. The corresponding success rates are reported in Table \ref{tab:res_conj}.

\begin{table}[!h]

\renewcommand{\arraystretch}{1.3}
\caption{Consistency rates for several mass function distances w.r.t. $\ocap$} 
\label{tab:res_conj}
\centering
\resizebox{.89\textwidth}{!}{
\begin{tabular}{|c||c|c|c|c|c|c|c|c|}
\hline
 Distance & $d_J$ & $d_{q,1}$ & $d_{q,2}$ & $d_{q,\infty}$ & $d_{pl,1}$ & $d_{pl,2}$ & $d_{pl,\infty}$ & $d_{\text{spe}}$\\
\hline
Consistency rate  & $86.42\%$ & $100\%$ & $100\%$ & $100\%$ & $38.22\%$ & $63.60\%$ & $100\%$ & $100\%$ \\
\hline
\end{tabular}}
\end{table}

The results are compliant with proposition \ref{prop:q_dist_conj} and \ref{prop:pl_dist_conj} as all commonality distances and $d_{pl,\infty}$ achieve $100\%$ of success. The results also show that Jousselme distance and $L_1$ or $L_2$ norm based distances between plausibilities are not consistent with $\ocap$. The rates also show that the circumstances in which Lipschitz continuity does not hold for these distances are not rare events.

To get a better insight as to why $d_{pl,k}$ is not consistent with $\ocap$ when $k$ is finite, we provide the following counter-example:
\begin{ex}
Let $\Omega = \left\{ a,b,c \right\}$. Suppose $m_1 = \frac{1}{2} m_{\left\{ a,b \right\}} + \frac{1}{2} m_\Omega$, $m_2 = \frac{1}{2} m_{\left\{ a,c \right\}} + \frac{1}{2} m_\Omega$ and $m_3 = m_{\left\{ b \right\}}$. By conjunctive combination, we obtain
\begin{align}
m_{1} \ocap m_{3} &= m_{\left\{ b \right\}} ,\\
\text{and } m_{2} \ocap m_{3} &= \frac{1}{2} m_{\left\{ b \right\}} + \frac{1}{2} m_\emptyset.
\end{align}
The plausibilities are 
\begin{center}
\begin{tabular}{ccccccccc}
& $\emptyset$ & $\left\{ a \right\}$ & $\left\{ b \right\}$ & $\left\{ a,b \right\}$ & $\left\{ c \right\}$ & $\left\{ a,c \right\}$ & $\left\{ b,c \right\}$ & $\Omega$ \\ \hline
$pl_1$ & 0 & 1 & 1 & 1 & $\frac{1}{2}$ & 1 & 1 & 1 \\
$pl_2$ & 0 & 1 & $\frac{1}{2}$ & 1 & 1 & 1 & 1 & 1 \\
$pl_{1\cap 3}$ & 0 & 0 & 1 & 1 & 0 & 0 & 1 & 1 \\
$pl_{2\cap 3}$ & 0 & 0 & $\frac{1}{2}$ & $\frac{1}{2}$ & 0 & 0 & $\frac{1}{2}$ & $\frac{1}{2}$ 
\end{tabular}
\end{center}
We see that 
\begin{align}
d_{pl,k}\left( m_1,m_2 \right) &= \frac{1}{\rho} \left( 2 \left( \frac{1}{2} \right)^k  \right)^{1/k} \\
\text{while }  d_{pl,k}\left( m_{1} \ocap m_{3},m_{2} \ocap m_{3} \right) &= \frac{1}{\rho} \left( 4 \left( \frac{1}{2} \right)^k  \right)^{1/k}.
\end{align}
This counter-example tends to show that the inconsistency of these distances lies (at least partially) in the way that the mass of the empty set is assigned by the conjunctive rule.
\end{ex}



\subsection{Main results on Lipschitz continuity for the disjunctive rule} 
\label{sub:main_results_on_consistent_distances_w_r_t_the_disjunctive_rule}


\begin{prop}\label{prop:pl_dist_disj}
 For $1 \leq k \leq \infty$, $\ocup$ is 1-Lipschitz continuous w.r.t. the $L_k$ norm based $pl$-distance $d_{pl,k}$.
\end{prop}

See \ref{ap:pl_dist_disj} for proof.

\begin{prop}\label{prop:q_dist_disj}
 $\ocup$ is 1-Lipschitz continuous w.r.t the $L_\infty$ norm based $q$-distance $d_{q,\infty}$.
\end{prop}

See \ref{ap:q_dist_disj} for proof.

The same type of arguments outlining the added value of our new Lipschitz continuity results in the conjunctive case also hold in the disjunctive one. The distances between plausibilities or commonalities have a smaller time and memory complexities as compared to the specialization distances. It must be noted that $d_{\text{spe},1}$, $d_{\text{spe},\infty}$, $d_{q,\infty}$ and $d_{pl,\infty}$ are the only distances that are reported to be consistent with both the conjunctive and disjunctive rules. As the numerical experiments presented in the next paragraph will show, $L_k$ norm based distances between commonalities are not consistent with $\ocup$ when $k$ is finite. The counter-example presented in  \ref{subsec:consist_exp} proves that $L_k$ norm based distances between plausibilities are not consistent with $\ocap$ when $k$ is finite.

\subsection{Lipschitz continuity for the disjunctive rule: numerical experiments and counter-examples}\label{subsec:consist_exp_disj}
This subsection contains experiments illustrating the (in)consistency of several mass function distances with respect to the disjunctive rule. We generate randomly \cite{bur13} $1e4$ triplets of simple mass functions and check how many times inequality \eqref{eq:consit} is verified (when $*=\ocup$) for several distances. The corresponding success rates are reported in Table \ref{tab:res_disj}.

\begin{table}[!h]

\renewcommand{\arraystretch}{1.3}
\caption{Consistency rates for several mass function distances w.r.t. $\ocup$} 
\label{tab:res_disj}
\centering
\resizebox{.89\textwidth}{!}{
\begin{tabular}{|c||c|c|c|c|c|c|c|c|}
\hline
 Distance & $d_J$ & $d_{q,1}$ & $d_{q,2}$ & $d_{q,\infty}$ & $d_{pl,1}$ & $d_{pl,2}$ & $d_{pl,\infty}$ & $d_{\text{spe}}$\\
\hline
Consistency rate  & $100\%$ & $94.76\%$ & $94.09\%$ & $100\%$ & $100\%$ & $100\%$ & $100\%$ & $100\%$ \\
\hline
\end{tabular}}
\end{table}

The results are compliant with proposition \ref{prop:pl_dist_disj} and \ref{prop:q_dist_disj} as all plausibility distances and $d_{q,\infty}$ achieve $100\%$ of success. The results also show that $L_1$ or $L_2$ norm based distances between commonalities are not consistent with $\ocup$. The consistency of Jousselme distance can be conjectured. This distance also achieves $100\%$ of success if one draws random mass functions and not just random simple mass functions.

To get a better insight as to why $d_{q,k}$ is not consistent with $\ocup$ when $k$ is finite, we provide the following counter-example:
\begin{ex}
Let $\Omega = \left\{ a,b,c \right\}$. Suppose $m_1 =  m_{\left\{ a \right\}} $, $m_2 =  m_{\left\{ a,c \right\}} $ and $m_3 = m_{\left\{ b \right\}}$. By disjunctive combination, we obtain
\begin{align}
m_{1} \ocup m_{3} &= m_{\left\{ a,b \right\}} ,\\
\text{and } m_{2} \ocup m_{3} &=  m_\Omega.
\end{align}
The commonalities are 
\begin{center}
\begin{tabular}{ccccccccc}
& $\emptyset$ & $\left\{ a \right\}$ & $\left\{ b \right\}$ & $\left\{ a,b \right\}$ & $\left\{ c \right\}$ & $\left\{ a,c \right\}$ & $\left\{ b,c \right\}$ & $\Omega$ \\ \hline
$q_1$ & 1 & 1 & 0 & 0 & 0 & 0 & 0 & 0 \\
$q_2$ & 1 & 1 & 0 & 0 & 1 & 1 & 0 & 0 \\
$q_{1\cup 3}$ & 1 & 1 & 1 & 1 & 0 & 0 & 0 & 0 \\
$q_{2\cup 3}$ & 1 & 1 & 1 & 1 & 1 & 1 & 1 & 1 
\end{tabular}
\end{center}
We see that 
\begin{align}
d_{q,k}\left( m_1,m_2 \right) &= \frac{1}{\rho} \left( 2   \right)^{1/k} \\
\text{while }  d_{q,k}\left( m_{1} \ocup m_{3},m_{2} \ocup m_{3} \right) &= \frac{1}{\rho} \left( 4   \right)^{1/k}.
\end{align}

\end{ex}

\section{Conflict degrees spanned by consistent distances with the conjunctive rule} 
\label{sec:conflict}

When information sources support antagonistic assumptions, it is important to provide a way to numerically assess the level of inconsistency in their respective messages. This is the purpose of degrees of conflict defined in the framework of belief functions. If such a degree of conflict is bounded, we can use different information fusion strategies in order to make more robust decisions. 

In the theory of belief functions, such a situation typically occurs when there is a pair of subsets $\left( A,B \right) $ such that $A\cap B = \emptyset$ and $m_1 \left( A \right)>0 $ and $m_2 \left( B \right) >0 $. 
In the following paragraphs, we give a brief reminder of existing conflict degrees in the belief function literature as well as desirable properties for such degrees. Next, we also comment on the advisability of building new degrees using distances that are consistent with $\ocap$.


\subsection{Assessing the degree of conflict between belief functions} 
\label{sub:assessing_the_degree_of_conflict}

In his pioneering article, Dempster \cite{Dem67} already provides a way to assess the degree of conflict between two mass functions. Let $\kappa$ denote this criterion which is known as \textbf{ Dempster's degree of conflict} and reads
\begin{equation}
	\kappa \left( m_1,m_2 \right) = m_{1 \cap 2} \left( \emptyset \right).  
\end{equation}

More recently, Destercke and Burger \cite{des2013} outline that this degree can be built upon a consistency measure $\phi$ which evaluates to what extent a single mass function is not self-contradictory. In the case of Dempster's degree of conflict, this measure is simply given by 
\begin{equation}
	\phi \left( m \right) = 1 - m \left( \emptyset \right).
\end{equation}

They also introduce the following strong consistency measure $\Phi$ which is such that
\begin{equation}
	\Phi \left( m \right) = \underset{a\in \Omega}{\max}\; pl \left( \left\{ a \right\} \right).
\end{equation}
This second measure is the $L_\infty$ norm of the contour function\footnote{The contour function is the restriction of the plausibility function to singletons.}. It is stronger in the sense that $\phi \left( m \right)<1 \Rightarrow \Phi \left( m \right) <1 $ while the opposite implication is not true. This means that $\Phi$ can detect a finer level of self-(in)consistencies in the information encoded by some mass function. 

Destercke and Burger \cite{des2013} also introduce the following list of desirable properties for some degree of conflict $\mathcal{C}$:
\begin{itemize}
 	\item [(i)] (extreme conflict values) $\mathcal{C} \left( m_1 , m_2 \right) = 0 $ iff $m_1$ and $m_2$ are non conflicting and $\mathcal{C} \left( m_1 , m_2 \right) = 1 $ iff $m_{1} \ocap m_2 = m_\emptyset$,
 	\item [(ii)] (symmetry) $\mathcal{C} \left( m_1 , m_2 \right) = \mathcal{C} \left( m_2 , m_1 \right) $,
 	\item [(iii)] (imprecision monotonicity) if $m_1 \sqsubseteq m'_1$ then $\mathcal{C} \left( m'_1 , m_2 \right) \leq \mathcal{C} \left( m_1 , m_2 \right) $,
 	\item [(iv)] (ignorance is bliss) $\mathcal{C} \left( m_1 , m_\Omega \right) = 1 - \mathfrak{I}\left( m_1 \right) $ where $\mathfrak{I}$ is a consistency measure such as the aforementioned ones,
 	\item [(v)] (invariance to refinement) for some multi-valued mapping $\rho : \Omega \rightarrow 2^\Theta$ with $|\Omega| < |\Theta| < \infty$ and a mass function $m'_1$ such that $m'_1 \left( \underset{a \in A}{\cup} \rho \left( a \right)  \right) = m_1 \left( A \right)  $ for any $A\subseteq \Omega$, we have $\mathcal{C} \left( m_1 , m_2 \right) = \mathcal{C} \left( m'_1 , m'_2 \right)$.
 \end{itemize} 

The definition of non-conflicting mass functions is not specified in property (i) because several such notions can be considered. The authors explain that if non-conflict means that the intersection of any focal element of $m_1$ with any focal element of $m_2$ is not empty then $\kappa$ satisfies each property with $\mathfrak{I} = \phi$. Moreover, if non-conflict means that the intersection of all the focal elements of both mass functions is not empty then $K \left( m_1,m_2 \right) = 1 - \underset{a\in \Omega}{\max}\; pl_{1\cap 2} \left( \left\{ a \right\} \right)  $ satisfies each property with $\mathfrak{I} = \Phi$. We will refer to $K$ as the \textbf{ degree of strong conflict}. The informational partial order in property (iii) is the specialization partial order \cite{Dub86b} for both degrees $\kappa$ and $K$.


\subsection{Deriving new degrees of conflict} 
\label{sub:deriving_new_degrees_of_conflict}

Prior to Destercke and Burger \cite{des2013}, several authors \cite{Mar08,LIU2006} proposed to derive new degrees of conflict to overcome the limitations of $\kappa$. Indeed, Dempster's degree of conflict evaluates two pairs of mass functions as equally conflicting as long as they assign the same mass to $\emptyset$ (after their respective conjunctive combinations). Let $m_\cap$ denote the conjunctive combination of the first pair and $m'_\cap$ the combination of the second one. Suppose the focal elements of $m_\cap$ are $\left\{ \emptyset,A \right\}$ and those of $m'_\cap$ are $\left\{ \emptyset,A,B \right\}$. If $A\cap B = \emptyset$, then $m'_\cap$ carries intuitively a higher level of inconsistency which $\kappa$ fails to grasp.

The degrees of conflict introduced in \cite{Mar08,LIU2006} are built using pairwise distances $d \left( m_1,m_2 \right) $. There are however several arguments \cite{des2013,burger2016geometric} outlining that this practice is ill advised. However, in a recent work, Pichon and Jousselme \cite{pichon2019several} highlighted that non-pairwise distances can be instrumental in the construction of degrees of conflict. The authors examine the distance between the conjunctive combination $m_{1\cap2}$ and some reference mass function, i.e. the total conflict mass function $m_\emptyset$. Indeed we have

\begin{align}
\kappa \left( m_1,m_2 \right) &= 1 - d_{pl,\infty} \left( m_{1\cap 2}, m_\emptyset \right). \label{eq:kappa_pl_inf}
\end{align}
Similarly, $K$ can be retrieved as the $L_\infty$ norm based distances between the contour functions of $m_1$ and $m_2$. This observation raises the following question: can we build other relevant degrees of conflict in the same fashion as in \eqref{eq:kappa_pl_inf} but using other distances than $d_{pl,\infty}$? We try to provide some answers to this question in the next paragraphs when the examined distances are consistent with $\ocap$.

\begin{prop}\label{rmk:1}
Let $d$ denote a mass function distance which is either an $L_k$ norm based distance between commonalities, plausibilities, or specialization matrices. Let $\mathcal{C} : \mathcal{M} \times \mathcal{M} \rightarrow \left[ 0;1 \right] $ denote the following mapping
\begin{equation}
	\mathcal{C} \left( m_1,m_2 \right)  = 1 - d \left( m_{1\cap 2},m_\emptyset \right). \label{eq:conf_deg}
\end{equation}
Then $\mathcal{C}$ does not satisfy property (i) if $k$ is finite.
\end{prop}
See \ref{sec:proof_of_proposition_rmk:1} for a proof.

From the above result, building a conflict degree using \eqref{eq:conf_deg} using $d_{q,k}$,  $d_{pl,k}$ or $d_{\text{spe},k}$ is ill-advised whenever $k\neq \infty$. Intuitively, degrees of conflict relying on $L_\infty$ are better candidates to verify (i) because the maximal norm value is not uniquely achieved for $m_{1\cap 2} = m_\Omega$.

\begin{prop}\label{rmk:2}
Let $\mathcal{C}_{q,\infty} : \mathcal{M} \times \mathcal{M} \rightarrow \left[ 0;1 \right] $ denote the following mapping
\begin{equation}
	\mathcal{C} _{q,\infty}\left( m_1,m_2 \right)  = 1 - d_{q,\infty} \left( m_{1\cap 2},m_\emptyset \right). \label{eq:conf_deg-q_inf}
\end{equation}
Then $\mathcal{C}_{q,\infty}$ coincides with the degree of strong conflict $K$.
\end{prop}

\begin{prop}\label{rmk:2_bis}
Let $\mathcal{C}_{\text{spe},\infty} : \mathcal{M} \times \mathcal{M} \rightarrow \left[ 0;1 \right] $ denote the following mapping
\begin{equation}
	\mathcal{C} _{\text{spe},\infty}\left( m_1,m_2 \right)  = 1 - d_{\text{spe},\infty} \left( m_{1\cap 2},m_\emptyset \right). \label{eq:conf_deg-q_inf}
\end{equation}
Then $\mathcal{C}_{\text{spe},\infty}$ coincides with Dempster's degree of conflict $\kappa$.

\end{prop}

See \ref{sec:proof_of_proposition_rmk:2} and \ref{sec:proof_of_proposition_rmk:2_bis} for proofs.

Proposition \ref{rmk:2} shows that the degree of strong conflict is retrieved through the $L_\infty$ norm based commonality distance while remark \ref{rmk:2_bis} shows that Dempster's degree of conflict is retrieved through the $L_\infty$ norm based specialization distance which are complementary observations in line with \cite{pichon2019several}. We continue with another more general remark, i.e. outside the sole scope of a given family of mass function distances.

\begin{prop}\label{rmk:3}
Let $d$ denote a mass function distance which is consistent w.r.t. $\ocap$. Let $\mathcal{C} : \mathcal{M} \times \mathcal{M} \rightarrow \left[ 0;1 \right] $ denote the mapping defined from \eqref{eq:conf_deg}.
Then $\mathcal{C}$ satisfies property (iii) w.r.t. the Dempsterian partial order $\sqsubseteq_d$.
\end{prop}

See \ref{sec:proof_of_proposition_rmk:3} for a proof. 
Following proposition \ref{rmk:3}, it seems that in general, distances consistent w.r.t. $\ocap$ are good candidates to possibly yield a relevant conflict degree.

\section{Conclusion}
\label{sec:ccl}

In the scope of the theory of belief functions, this paper provides new results on the consistency of $L_k$ norm based distances between commonalities and the $L_\infty$ norm based distance between plausibilities with the conjunctive rule of combination. We also prove the consistency of $L_k$ norm based distances between plausibilities and the $L_\infty$ norm based distance between commonalities with the disjunctive rule of combination. The investigated form of consistency is equivalent to Lipschitz continuity of the mapping obtained by fixing one of the operand of pairwise combinations under these rules. Since the corresponding Lipschitz constant is 1, this property means that combining any pair of belief functions with any third party belief function is a non-expansive operation. 

Outside the scope of belief functions, the results apply to random set distributions as belief functions can be interpreted as epistemic random sets. In this more general context, the conjunctive rule yields the distribution of the intersection of two independent random sets while the disjunctive rule yields the distribution of the union of two independent random sets. Commonalities map any subset $A$ to the probability that the random set is a superset of $A$ (inclusion functional). Plausibilities map any subset $A$ to the probability that the random set intersects $A$ (hitting functional). Our results prove that if $F$ maps the distribution of a random set to the distribution of the intersection of this random set with a given (fixed) independent one, then $F$ is Lipschitz continuous with Lipschitz constant 1 w.r.t. $L_k$ norm based distances between inclusion functionals or the $L_\infty$ norm based distance between hitting functionals. Similarly, if $F$ maps the distribution of a random set to the distribution of the union of this random set with a given (fixed) independent one, then $F$ is Lipschitz continuous with Lipschitz constant 1 w.r.t. $L_k$ norm based distances between hitting functionals or the $L_\infty$ norm based distance between inclusion functionals.

We only investigate belief functions and random sets on finite spaces. Extending these results to uncountable spaces is an important perspective for future works. In the uncountable setting, random closed sets are defined as measurable mappings with respect to the Effros $\sigma$-algebra on the family of closed subsets of some locally compact Haussdorf completely separable topological space. The main results obtained in the finite case essentially rely on two aspects:
\begin{itemize}
	\item[(i)] intersection (resp. union) of independent random sets can be characterized by the elementwise multiplication of their inclusion (resp. containment) functionals,
	\item[(ii)] a closed form expression of $L_k$ norms for inclusion or hitting functionals.
\end{itemize}
Concerning the first aspect, 
the fact that $(A\supseteq C \text{ and } B \supseteq C) \Leftrightarrow A\cap B  \supseteq C$ and $(A\subseteq C \text{ and } B \subseteq C) \Leftrightarrow A\cup B  \subseteq C$ is intuitively sufficient to obtain equivalent relations in uncountable spaces. The relation for unions of independent random closed sets is indirectly evoked for containment functionals in \cite[p.82]{molchanov2005theory}. 
The second aspects seems more challenging to generalize because one needs to introduce a norm for capacity functionals that is not just a vector norm. It may be possible to build such norms using Choquet integrals.

Another relevant research track for future works consists in investigating if the proposed Liptschitz continuity results hold as well when independence assumptions are not verified. Intuitively intersecting or uniting some pair of random sets with a third one that may or may not be dependent on either of them should still make their corresponding distributions closer (w.r.t. the appropriate distance). The difficulty in this regard is that one can no longer just work with marginal inclusion or containment functionals but with multivariate ones \cite{schmelzer2012characterizing} that do not factorize as elementwise products of marginal functionals. Recent work \cite{schmelzer2015joint,schmelzer2015sklar,schmelzer2018multivariate} generalizing Sklar's theorem on copulas to joint or multivariate capacity functionals of random sets may be useful in this quest because it gives an explicit connection between marginal functionals and multivariate ones. However, in contrast to point-valued random variables, a family of copulas is necessary to characterize this link.

Finally, going back to belief functions, we also discuss the advisability of building new degrees of conflict using distances that are consistent with the conjunctive rule. Such degrees are mainly interesting in information fusion applications of belief functions. We show that distances consistent with the conjunctive rule can be deemed to be relevant candidates for this purpose as they will achieve a desirable property for degrees of conflict. As for the distances for which we provide new consistency results (distance between commonalities or plausibilities), it turns out that they either violate another desirable property or coincide with an already known degree of conflict.

\section*{Acknowledgments}
The author is indebted to Frederic Pichon for the useful discussions on degrees of conflict and his remarks that helped him clarifying the meaning of the consistency property.

\appendix

\section{Set metrics and their consistency with set operations}
\label{ap:set_dist}

This appendix is meant to show that some distances between sets verify either property (a) or (b), see subsection \ref{sub:a_definition_of_consistency_between_distances_and_combinations_rules} for the definitions of these latter. We examine, the Hamming distance and the Jaccard distance:
\begin{itemize}
	\item Hamming distance: for any subsets $A,B$ and $C$, we have 
	\begin{align}
	d_{\text{ham}} \left( A\cap C,B\cap C \right) &= \frac{ |\left( A\cap C \right) \Delta \left( B \cap C \right)  |}{n} \\
	&= \frac{ |\left( A \Delta B \right) \cap C |}{n} \\
	&\leq \frac{ | A \Delta B  |}{n}.
	\end{align}
	So we see that $d_{\text{ham}}$ verifies (a).

	We can also write
	\begin{align}
	d_{\text{ham}} \left( A\cup C,B\cup C \right) &= \frac{ |\left( A\cup C \right) \Delta \left( B \cup C \right)  |}{n} \\
	&= \frac{ |\left( A \cup B \cup C  \right) \setminus \left( \left( A\cup C \right) \cap \left( B\cup C \right)   \right)  |}{n} \\
	&= \frac{ |\left( A \cup B \cup C  \right) \setminus \left( \left( A\cap B \right) \cup C   \right)  |}{n} \\
	&= \frac{ |\left( \left(   A\setminus C \right) \cup \left(   B \setminus C  \right) \right) \setminus \left(  A\cap B \right)   |}{n} \\
	&\leq \frac{ |\left( A \cup    B  \right) \setminus \left(  A\cap B \right)   |}{n} \\
	&\leq \frac{ | A \Delta B  |}{n},
	\end{align}
	and $d_{\text{ham}}$ verifies (b) as well.
	\item Jaccard distance: for any subsets $A,B$ and $C$, we have
	\begin{align}
	d_{\text{jac}} \left( A\cup C,B\cup C \right) &= \frac{|\left( A\cup C \right) \Delta \left( B \cup C \right)  |}{ | \left( A\cup C \right) \cup \left( B\cup C \right) | }\\
	&= \frac{| A \Delta B  | - | (A \Delta B) \setminus C |}{|A\cup B| + |C \setminus \left( A \cup B \right) |} \\
	&\leq \frac{| A \Delta B  |}{|A\cup B|}
	\end{align}
	So we see that $d_{\text{jac}}$ verifies (b).
	
To see that $d_{\text{jac}}$ does not verify property (a), we provide the following counter-example.

\begin{ex}
Let $A$ and $B$ denote two subsets that are not disjoint, i.e. $A\cap B \neq \emptyset$ and therefore $d_{\text{jac}} \left( A,B \right) = 1 - \frac{|A \cap B|}{|A\cup B|} < 1$.

If $C = A \Delta B$, then $A\cap C = A\setminus B$ and $B \cap C = B \setminus A$. This also implies that $\left( A\cap C \right) \Delta \left( B \cap C \right) = A \Delta B$ because $A\cap C$ and $B \cap C$ are disjoint. We also have $\left( A\cap C \right) \cup \left( B \cap C \right) = A \Delta B$ and consequently $d_{\text{jac}} \left( A\cap C,B \cap C \right) = 1$.
\end{ex}
	
\end{itemize}

\section{Proof of lemma \ref{lem:implic}} 
\label{ap:proof_of_lemma_lem:implic}

In this appendix, we give a proof that $b_{1|E} \left( A \right) = b_1 \left( \left( E\setminus A \right)^c  \right) $ for any implicability function $b_1$ and any subset $A$ and $E$.

\begin{proof}
By definition of the implicability function and conditioning, we have for any $A\subseteq \Omega$
\begin{align}
b_{1|E} \left( A \right) &= \sum_{B\subseteq A} \sum_{\substack{C\subseteq \Omega \\ \text{s.t.} \\ C\cap E = B}} m_1 \left( C \right). 
\end{align}
The second sum is empty if $E$ is not a superset of $B$. If this condition is verified, we remark that subsets $C$ are necessarily the union of $B$ and some subset of $E^c$. This gives
\begin{align}
b_{1|E} \left( A \right) &= \sum_{B\subseteq A \cap E} \sum_{D\subseteq E^c } m_1 \left(B \cup D \right). 
\end{align}
Finally, any subset $X$ of $\left( E\setminus A \right)^c $ can be partitioned w.r.t. $A\cap E$ and $E^c$, meaning that $\exists ! Y\subseteq A\cap E$ and $\exists ! Y' \subseteq E^c$ such that $X = Y\cup Y'$. Consequently, we have 
\begin{align}
b_{1|E} \left( A \right) &= \sum_{X \subseteq \left( E\setminus A \right)^c } m_1 \left( X \right) \\
&= b_1 \left( \left( E\setminus A \right)^c  \right) 
\end{align}

\end{proof}


\section{Proof of proposition \ref{prop:q_dist_conj}} 
\label{ap:q_dist_conj}

In this appendix, we give a proof that for $1 \leq k \leq \infty$, $\ocap$ is 1-Lipschitz continuous w.r.t. the $L_k$ norm based $q$-distance $d_{q,k}$.

\begin{proof}
 Suppose $m_1$, $m_2$ and $m_3$ are three mass functions on $\Omega$ and $\mathbf{q}_1$, $\mathbf{q}_2$ and $\mathbf{q}_3$ are their respective commonality vectors. For any positive finite integer $k$, we have:
\begin{eqnarray}
 \left[  d_{q,k}\pare{m_1\ocap m_3, m_2\ocap m_3} \right]^k &=& \left[   \left\|\mathbf{q}_{1\cap 3} - \mathbf{q}_{2\cap 3} \right\|_{k}\right]^k \nonumber\\
   &=& \sum_{A\subseteq \Omega} |q_{1\cap 3}\pare{A} - q_{2\cap 3}\pare{A}|^k \nonumber \\
   &=& \sum_{A\subseteq \Omega} |q_{1}\pare{A}q_{3}\pare{A} - q_{2}\pare{A}q_{3}\pare{A}|^k \nonumber\\
   &=& \sum_{A\subseteq \Omega} |q_{3}\pare{A}|^k |q_{1}\pare{A} - q_{2}\pare{A}|^k, \nonumber
\end{eqnarray}
For any subset $A\subseteq \Omega$, we have that $0 \leq q_{3}\pare{A}\leq 1$ thus we obtain
\begin{eqnarray}
 \left[  d_{q,k}\pare{m_1\ocap m_3, m_2\ocap m_3} \right]^k &\leq& \sum_{A\subseteq \Omega}  |q_{1}\pare{A} - q_{2}\pare{A}|^k  \nonumber\\
   &\leq& \left[  d_{q,k}\pare{m_1, m_2} \right]^k.\nonumber
\end{eqnarray}
By definition, this latter inequality means that distance $d_{q,k}$ is consistent with rule $\ocap$.

If $k=\infty$, we have:
\begin{eqnarray}
 d_{q,\infty}\pare{m_1\ocap m_3, m_2\ocap m_3} &=& \left\|\mathbf{q}_{1\cap 3} - \mathbf{q}_{2\cap 3} \right\|_{\infty} \nonumber\\
   &=& \max_{A\subseteq \Omega} |q_{1\cap 3}\pare{A} - q_{2\cap 3}\pare{A}| \nonumber \\
   &=& \max_{A\subseteq \Omega} |q_{1}\pare{A}q_{3}\pare{A} - q_{2}\pare{A}q_{3}\pare{A}|  \nonumber\\
   &=& q_{3}\pare{B} |q_{1}\pare{B} - q_{2}\pare{B}|, \nonumber
\end{eqnarray}
with $B=\underset{A\subseteq \Omega}{\arg\max} \braces{q_{3}\pare{A} |q_{1}\pare{A} - q_{2}\pare{A}|}$. It follows that
\begin{eqnarray}  
    d_{q,\infty}\pare{m_1\ocap m_3, m_2\ocap m_3} &\leq & |q_{1}\pare{B} - q_{2}\pare{B}|  \nonumber\\
   &\leq & \max_{A\subseteq \Omega} |q_{1}\pare{A} - q_{2}\pare{A}|  \nonumber\\
   &\leq & d_{q,\infty}\pare{m_1, m_2}.\nonumber
\end{eqnarray}
By definition, this latter inequality means that $\ocap$ is 1-Lipschitz continuous w.r.t. $d_{q,\infty}$.
\end{proof}


\section{Proof of proposition \ref{prop:pl_dist_conj}} 
\label{ap:pl_dist_conj}

In this appendix, we give a proof that $\ocap$ is 1-Lipschitz continuous w.r.t. the $L_\infty$ norm based $pl$-distance $d_{pl,\infty}$.

\begin{proof}
Suppose $m_1$, $m_2$ and $m_3$ are three mass functions on $\Omega$ and $\mathbf{pl}_1$, $\mathbf{pl}_2$ and $\mathbf{pl}_3$ are their respective plausibility vectors. Let us first prove an intermediate result in case $m_3=m_E$ is categorical. We can write
\begin{align}
 d_{pl,\infty}\left( m_{1|E}, m_{2|E} \right) &= \underset{A\subseteq \Omega}{\max}\; |pl_{1|E}\left( A \right) - pl_{2|E} \left( A \right)  |.
 \end{align}
 Using the fact that $pl$ and $b$-distances coincide and lemma \ref{lem:implic}, we obtain
 \begin{align}
 d_{pl,\infty}\left( m_{1|E}, m_{2|E} \right) &= \underset{A\subseteq \Omega}{\max}\; |b_{1|E}\left( A \right) - b_{2|E} \left( A \right)  | \\
 &= \underset{A\subseteq \Omega}{\max}\; |b_{1}\left( \left( E\setminus A \right)^c  \right) - b_{2} \left( \left( E\setminus A \right)^c \right)  | \\
 &\leq \underset{A\subseteq \Omega}{\max}\; |b_{1}\left( A  \right) - b_{2} \left( A \right)  |.
 \end{align}
 So the consistency condition is verified when $m_3$ is categorical. Now, let us examine the general case where $m_3$ is not necessarily categorical. Let $\mathbf{B}$ denote the transfer matrix \cite{Sme02} allowing to obtain vector forms of implicability functions by right-handed dot product with the vector form of their corresponding mass functions. We can write
 \begin{align}
 d_{pl,\infty}\left( m_{1}\ocap m_3, m_{2}\ocap m_3 \right) &= \left\| \mathbf{pl}_{1\cap3} - \mathbf{pl}_{2\cap3}\right\|_\infty \\
 &= \left\| \mathbf{b}_{1\cap3} - \mathbf{b}_{2\cap3}\right\|_\infty \\
 &= \left\| \mathbf{B} \cdot \left( \mathbf{m}_{1\cap3} - \mathbf{m}_{2\cap3} \right)  \right\|_\infty \\
 &= \left\| \mathbf{B} \cdot \left( \mathbf{S}_{1} - \mathbf{S}_{2} \right) \cdot \mathbf{m}_3  \right\|_\infty.
 \end{align}
 One can always decompose a mass function as a convex combination of categorical ones: $\mathbf{m}_3 = \sum\limits_{E \subseteq \Omega} m_3 \left( E \right)  \mathbf{m}_E$. We obtain
 
\begin{align}
d_{pl,\infty}\left( m_{1}\ocap m_3, m_{2}\ocap m_3  \right) & =  \left\| \sum_{E \subseteq \Omega} m_3 \left( E \right)\mathbf{B} \cdot \left( \mathbf{S}_{1} - \mathbf{S}_{2} \right) \cdot \mathbf{m}_E  \right\|_\infty \\
&\leq \sum_{E \subseteq \Omega} m_3 \left( E \right) \left\| \mathbf{B} \cdot \left( \mathbf{S}_{1} - \mathbf{S}_{2} \right) \cdot \mathbf{m}_E  \right\|_\infty,
\end{align}
with the last inequality following from the triangle inequality and absolute homogeneity properties of the $L_\infty$ norm.

From our intermediate result, we know that for any $E$, $\left\| \mathbf{B} \cdot \left( \mathbf{S}_{1} - \mathbf{S}_{2} \right) \cdot \mathbf{m}_E  \right\|_\infty = d_{pl,\infty} \left( m_{1|E},m_{2|E} \right) \leq d_{pl,\infty} \left( m_{1},m_{2} \right) $. Consequently, we have
\begin{align}
d_{pl,\infty}\left( m_{1}\ocap m_3, m_{2}\ocap m_3 \right) &\leq d_{pl,\infty} \left( m_{1},m_{2} \right)  \sum_{E \subseteq \Omega} m_3 \left( E \right) \\
&\leq d_{pl,\infty} \left( m_{1},m_{2} \right).
\end{align}

By definition, this latter inequality means that $\ocap$ is 1-Lipschitz continuous w.r.t. $d_{pl,\infty}$.
\end{proof}


\section{Proof of proposition \ref{prop:pl_dist_disj}} 
\label{ap:pl_dist_disj}

In this appendix, we give a proof that for $1 \leq k \leq \infty$, $\ocup$ is 1-Lipschitz continuous w.r.t. the $L_k$ norm based $pl$-distance $d_{pl,k}$.

\begin{proof}
Suppose $m_1$, $m_2$ and $m_3$ are three mass functions on $\Omega$ and $\mathbf{pl}_1$, $\mathbf{pl}_2$ and $\mathbf{pl}_3$ are their respective plausibility vectors. For any positive finite integer $k$, we have:
\begin{eqnarray}
 \left[  d_{pl,k}\pare{m_1\ocup m_3, m_2\ocup m_3} \right]^k &=& \left[   \left\|\mathbf{pl}_{1\cup 3} - \mathbf{pl}_{2\cup 3} \right\|_{k}\right]^k \nonumber\\
   &=& \sum_{A\subseteq \Omega} |pl_{1\cup 3}\pare{A} - pl_{2\cup 3}\pare{A}|^k \nonumber \\
   &=& \sum_{A\subseteq \Omega} |b_{1\cup 3}\pare{A^c} - b_{2\cup 3}\pare{A^c}|^k \nonumber \\
   &=& \sum_{A\subseteq \Omega} |b_{1}\pare{A^c}b_3 \left( A^c \right)  - b_{2}\pare{A^c}b_3 \left( A^c \right) |^k \nonumber \\
   &=& \sum_{A\subseteq \Omega} |b_3 \left( A^c \right)|^k |b_{1}\pare{A^c}  - b_{2}\pare{A^c} |^k. \nonumber
\end{eqnarray}
For any subset $A\subseteq \Omega$, we have that $0 \leq b_{3}\pare{A^c} \leq 1$ and we obtain
\begin{eqnarray}
 \left[  d_{pl,k}\pare{m_1\ocup m_3, m_2\ocup m_3} \right]^k  &\leq& \sum_{A\subseteq \Omega}  |b_{1}\pare{A^c}  - b_{2}\pare{A^c} |^k \nonumber \\
 &\leq& \sum_{A\subseteq \Omega}  |pl_{1}\pare{A}  - pl_{2}\pare{A} |^k\\ \nonumber
 &\leq& \left[  d_{pl,k}\pare{m_1, m_2} \right]^k.
\end{eqnarray}
By definition, this latter inequality means that distance $d_{pl,k}$ is consistent with rule $\ocup$.

If $k=\infty$, we have:
\begin{eqnarray}
 d_{pl,\infty}\pare{m_1\ocup m_3, m_2\ocup m_3} &=& \left\|\mathbf{pl}_{1\cup 3} - \mathbf{pl}_{2\cup 3} \right\|_{\infty}, \nonumber\\
   &=& \max_{A\subseteq \Omega} |pl_{1\cup 3}\pare{A} - pl_{2\cup 3}\pare{A}| \nonumber \\
   &=& \max_{A\subseteq \Omega} |b_{1\cup 3}\pare{A^c} - b_{2\cup 3}\pare{A^c}| \nonumber \\   
   &=& \max_{A\subseteq \Omega} |b_{1}\pare{A^c}b_{3}\pare{A^c} - b_{2}\pare{A^c}b_{3}\pare{A^c}| \nonumber\\
   &=& b_{3}\pare{B^c} |b_{1}\pare{B^c} - b_{2}\pare{B^c}|, \nonumber
\end{eqnarray}
with $B=\underset{A\subseteq \Omega}{\arg\max} \braces{b_{3}\pare{A^c} |b_{1}\pare{A^c} - b_{2}\pare{A^c}|}$. It follows that
\begin{eqnarray}  
    d_{pl,\infty}\pare{m_1\ocup m_3, m_2\ocup m_3} &\leq & |b_{1}\pare{B^c} - b_{2}\pare{B^c}| \nonumber\\
    &\leq & |pl_{1}\pare{B} - pl_{2}\pare{B}| \nonumber\\
   &\leq & \max_{A\subseteq \Omega} |pl_{1}\pare{A} - pl_{2}\pare{A}| \nonumber\\
   &\leq & d_{pl,\infty}\pare{m_1, m_2}.\nonumber
\end{eqnarray}
By definition, this latter inequality means that $\ocup$ is 1-Lipschitz continuous w.r.t. $d_{pl,\infty}$.

\end{proof}


\section{Proof of proposition \ref{prop:q_dist_disj}} 
\label{ap:q_dist_disj}

In this appendix, we give a proof that $\ocup$ is 1-Lipschitz continuous w.r.t the $L_\infty$ norm based $q$-distance $d_{q,\infty}$.

\begin{proof}
Suppose $m_1$, $m_2$ and $m_3$ are three mass functions on $\Omega$ and $q_1$, $q_2$ and $q_3$ are their respective commonality functions. Using relations \eqref{eq:q_and_b} and \eqref{eq:de_morgan}, we can write
\begin{align}
d_{q,\infty} \left( m_{1\cup 3}, m_{2\cup 3} \right) &= \underset{A\subseteq \Omega}{\max}\; |q_{1\cup3}\left( A \right) - q_{2\cup 3}\left( A \right)  |\\
&= \underset{A\subseteq \Omega}{\max}\; |\overline{b}_{1\cup3}\left( A^c \right) - \overline{b}_{2\cup 3}\left( A^c \right)  |\\
&= d_{b,\infty} \left( \overline{m_1 \ocup m_3} , \overline{m_2 \ocup m_3} \right)\\
&= d_{pl,\infty}\left( \overline{m_1 \ocup m_3} , \overline{m_2 \ocup m_3} \right)\\
&= d_{pl,\infty}\left( \overline{m}_{1}\ocap \overline{m}_{3} , \overline{m}_{2}\ocap \overline{m}_{3} \right).
\end{align}
Since $d_{pl,\infty}$ is consistent w.r.t. $\ocap$, we obtain
\begin{align}
d_{q,\infty} \left( m_{1\cup 3}, m_{2\cup 3} \right) &\leq d_{pl,\infty}\left( \overline{m}_{1} , \overline{m}_{2} \right)\\
\text{and } d_{pl,\infty}\left( \overline{m}_{1} , \overline{m}_{2} \right) &= d_{b,\infty}\left( \overline{m}_{1} , \overline{m}_{2} \right)\\
&=  \underset{A\subseteq \Omega}{\max}\; |\overline{b}_{1}\left( A \right) - \overline{b}_{2}\left( A \right)  |\\
&=  \underset{A\subseteq \Omega}{\max}\; |q_{1}\left( A^c \right) - q_{2}\left( A^c \right)  |\\
&=  d_{q,\infty} \left( m_{1}, m_{2} \right).
\end{align}
By definition, this latter inequality means that $\ocup$ is 1-Lipschitz continuous w.r.t. $d_{q,\infty}$.
\end{proof}


\section{Proof of proposition \ref{rmk:1}} 
\label{sec:proof_of_proposition_rmk:1}

In this appendix, we provide a proof that the extreme conflict values property cannot be verified when a degree of conflict $\mathcal{C}$ is defined as 
$$\mathcal{C} \left( m_1,m_2 \right) = 1 - d \left( m_{1\cap 2},m_\emptyset \right)  $$
where $d$ is an $L_k$ norm based distance between either commonality functions, plausibility functions or specialization matrices when $k$ is finite.

\begin{proof}
We obviously have $\mathcal{C} \left( m_1,m_2 \right)  = 1$ iff $m_{1\cap 2}= m_\emptyset$ and mass functions $m_1$ and $m_2$ are maximally conflicting so the problem does not come from this side of property (i). 

Now, suppose $d = d_{q,k}$ and $k$ is finite. Since $\mathcal{C} \left( m_1,m_2 \right)  = 0 \Leftrightarrow d_{q,k} \left( m_{1\cap 2},m_\emptyset \right)=1$, we can write

\begin{align}
 d_{q,k} \left( m_{1\cap 2},m_\emptyset \right) &= 1\\
 \Leftrightarrow \frac{1}{\rho} \left\| \mathbf{q}_{1\cap 2} - \mathbf{q}_{\emptyset} \right\|_k &= 1\\
 \Leftrightarrow \frac{1}{\rho^k} \sum_{A\subseteq \Omega} |q_{1\cap 2} \left( A \right)  - q_{\emptyset} \left( A \right) |^k &= 1.
\end{align}
Remember that $\rho^k = N-1$ for $L_k$ norm based commonality distances. Moreover, $q_\emptyset \left( A \right)=0$ when $A\neq \emptyset$ and that $q \left( \emptyset \right)=1 $ for any commonality function therefore we obtain
\begin{align}
 d_{q,k} \left( m_{1\cap 2},m_\emptyset \right) &= 1 \\
 \Leftrightarrow  \sum_{\substack{ A\subseteq \Omega \\ A \neq \emptyset}} |q_{1\cap 2} \left( A \right)  |^k &= N-1.
\end{align}
Finally, the sum of commonalities (power $k$) cannot be equal to $N-1$ unless $q_{1\cap 2} \left( A \right)=1$ for each $A\neq \emptyset$. This means that $q_{1\cap 2}= q_\Omega$ which implies that $q_1=q_2=q_\Omega$. This latter condition is not an admissible definition of non-conflict. The proof for distances between plausibilities is extremely similar and is thus omitted. 

Concerning, distances between specialization matrices, the philosophy is also similar but we provide a sketch of the proof. 
When dealing with an $L_k$ norm based specialization distance, we have $\rho^k = 2 \left( N-1 \right) $ and this distance is achieved for the pair $\left( m_\Omega, m_\emptyset \right) $. To see that $m_\Omega$ is the only mass function achieving maximal distance with $m_\emptyset$, one just needs to observe that the matrix entry $S_\emptyset(A,B)=1 $ if $A=\emptyset$ and $S_\emptyset(A,B)=0 $ otherwise. We can write
\begin{align}
\left[   d_{\text{spe},k}\left( m_{1\cap 2} ,m_\emptyset\right) \right]^k &= \sum_{E\subseteq \Omega} \left( \left\| \mathbf{m}_{1\cap 2|E} - \mathbf{m}_\emptyset \right\|_k \right)^k.
\end{align}
The only way to maximize the above expression is to maximize each $\left\| \mathbf{m}_{1\cap 2|E} - \mathbf{m}_\emptyset \right\|_k$ individually for $E\neq \emptyset$. We know that $\left\| \mathbf{m}_{1\cap 2|E} - \mathbf{m}_\emptyset \right\|_k\leq 2$ and we need $m_{1\cap 2|E} \left( \emptyset \right)=0 $ to achieve this maximal value. This is not possible unless $m_{1\cap 2} = m_\Omega$. Again, $m_{1\cap 2} = m_\Omega$ implies that $m_1 = m_2 = m_\Omega$.
\end{proof}


\section{Proof of proposition \ref{rmk:2}} 
\label{sec:proof_of_proposition_rmk:2}

In this appendix, we provide a proof that the degree of conflict $\mathcal{C}_{q,\infty}$ defined as 
$$\mathcal{C}_{q,\infty} \left( m_1,m_2 \right) = 1 - d_{q,\infty} \left( m_{1\cap 2},m_\emptyset \right)  $$
coincides with the degree of strong conflict $K$.

\begin{proof}
By definition of $\mathcal{C}_{q,\infty}$, we can write
\begin{align}
	\mathcal{C}_{q,\infty} \left( m_1,m_2 \right)  &=  1 - \underset{A \subseteq \Omega}{\max}\; |q_{1\cap 2} \left( A \right) - q_\emptyset \left( A \right) |.
	\end{align}
	Since $q_\emptyset \left( A \right)=0$ when $A\neq \emptyset$ and that $q \left( \emptyset \right)=1 $ for any commonality function, we obtain
	\begin{align}
	\mathcal{C}_{q,\infty} \left( m_1,m_2 \right) &= 1 - \underset{\substack {A \subseteq \Omega \\ A \neq \emptyset}}{\max}\; |q_{1\cap 2} \left( A \right)|.
	\end{align}	
	Any commonality function is such that $q \left( A \right) \geq q \left( A' \right)  $ if $A \subseteq A'$, which gives
	\begin{align}
	\mathcal{C}_{q,\infty} \left( m_1,m_2 \right) &= 1 - \underset{\substack {a \in \Omega \\ A \neq \emptyset}}{\max}\; |q_{1\cap 2} \left( \left\{ a \right\} \right)|.
	\end{align}	
	Finally commonality and plausibility functions coincide on singletons, hence $\mathcal{C}_{q,\infty} \left( m_1,m_2 \right) = K \left( m_1,m_2 \right)$.
\end{proof}


\section{Proof of proposition \ref{rmk:2_bis}} 
\label{sec:proof_of_proposition_rmk:2_bis}

In this appendix, we provide a proof that the degree of conflict $\mathcal{C}_{\text{spe},\infty}$ defined as 
$$\mathcal{C}_{\text{spe},\infty} \left( m_1,m_2 \right) = 1 - d_{\text{spe},\infty} \left( m_{1\cap 2},m_\emptyset \right)  $$
coincides with Dempster's degree of conflict $\kappa$.

\begin{proof}
By definition, we have
\begin{align}
d_{\text{spe},\infty} \left( m_{1\cap 2},m_\emptyset \right) &= \underset{A,B\subseteq \Omega}{\max}\; |S_{1\cap 2}\left( A,B \right) - S_\emptyset \left( A,B \right)  | \\
&= \underset{A\subseteq \Omega}{\max}\; \left\| \mathbf{m}_{1\cap 2|A} - \mathbf{m}_\emptyset \right\|_\infty.
\end{align}
For any $A\subseteq \Omega$, we have
\begin{align}
\left\| \mathbf{m}_{1\cap 2|A} - \mathbf{m}_\emptyset \right\|_\infty &= \max \left\{ 1-m_{1\cap 2|A}\left( \emptyset \right);   \underset{\substack{E\subseteq \Omega \\ E\neq \emptyset}}{\max}\; m_{1\cap 2|A}\left( E \right) \right\}.
\end{align}
For any $E\neq \emptyset$, we have
\begin{align}
m_{1\cap 2|A}\left( E \right) &\leq \sum_{\substack{E'\subseteq \Omega \\ E'\neq \emptyset}} m_{1\cap 2|A}\left( E' \right) \\
&\leq 1 - m_{1\cap 2|A}\left( \emptyset \right).
\end{align}

We deduce that $\left\| \mathbf{m}_{1\cap 2|A} - \mathbf{m}_\emptyset \right\|_\infty = 1-m_{1\cap 2|A}\left( \emptyset \right)$. Finally, Dempster's degree of conflict can only grow as one performs a conjunctive combination therefore $\underset{A\subseteq \Omega}{\max}\;1-m_{1\cap 2|A}\left( \emptyset \right) = 1 -m_{1\cap 2}\left( \emptyset \right)  $, hence $\mathcal{C} _{\text{spe},\infty}\left( m_1,m_2 \right) = \kappa \left( m_1,m_2 \right) $.
\end{proof}

\section{Proof of proposition \ref{rmk:3}} 
\label{sec:proof_of_proposition_rmk:3}

In this appendix, we give a proof that for some degree of conflict $\mathcal{C}$ defined as 
$$\mathcal{C} \left( m_1,m_2 \right) = 1 - d \left( m_{1\cap 2},m_\emptyset \right),  $$
where $d$ is a mass function distance consistent with $\ocap$, then $\mathcal{C}$ satisfies property (iii) (from \cite{des2013}) w.r.t. the Dempsterian partial order $\sqsubseteq_d$.

\begin{proof}
By definition of the Dempsterian partial order, $m_1 \sqsubseteq_d m_1'$ means that there exists a mass function $m$ such that $m_1 = m'_1 \ocap m$. 
If $d$ is consistent w.r.t. $\ocap$, then for any mass function  $m_2$, we have 
\begin{align}
d \left( m'_1 \ocap m_2 \ocap m, m_\emptyset \ocap m \right) &\leq  d \left( m'_1 \ocap m_2 , m_\emptyset  \right) \\
\Leftrightarrow d \left( m_1 \ocap m_2 , m_\emptyset \right) &\leq  d \left( m'_1 \ocap m_2 , m_\emptyset  \right) \\
\Leftrightarrow \mathcal{C} \left( m_1,m_2 \right) &\geq \mathcal{C} \left( m'_1,m_2 \right).
\end{align}

\end{proof}




\section{Consistency of distances with $\alpha$-junctions}\label{ap:alpha_junction}

In \cite{Lou16}, two families of mass function distances are introduced. Each of them relies on a given type of evidential matrix and a matrix norm. The evidential matrices in question are either an $\alpha$-specialization matrix or an $\alpha$-generalization matrix. These matrices are a generalization of the specialization matrix and generalization matrix, in the sense that these two matrices are retrieved by setting $\alpha=1$. The definition of these more general matrices stems from a class of combination rules known as $\alpha$-junctions \cite{Sme97}. 

In short, $\alpha$-junctions are linear combination rules that do not depend on the order in which mass functions are combined. This axiomatic justification of these properties is detailed in~\cite{Sme97}. These rules also have a meta-data dependent interpretation. These meta-data characterize the truthfulness of the sources that induced the mass functions. A source is truthful if it conveys the pieces of information it possesses and it is untruthful if it conveys inconsistent pieces of information as compared to the ones it possesses. For example, suppose a source has inferred that $\left\{ \theta \in A \right\}$. If it is truthful it conveys the mass function $m_A$ while it conveys $m_{A^c}$ if it is untruthful. The $\alpha$-junctions allow to combine mass functions in several situations ranging between these two extreme cases. This interpretation is documented in \cite{Pic09,Pic12-2,Lou16,Kle14}.

Concerning evidential matrices, the most important point for the present discussion is that for each $\alpha \in \left[ 0;1 \right] $, there is a bijective correspondence between a given mass function $m_i$ and an $\alpha$-specialization matrix $\mathbf{D}^{(\alpha)}_i$. Another bijective correspondence also exists between a given mass function $m_i$ and an $\alpha$-generalization matrix $\mathbf{G}^{(\alpha)}_i$. The main result in \cite{Lou16} states that, for any $\alpha \in \left[ 0;1 \right] $, the distance induced by the $L_1$ matrix norm of the difference between a pair of $\alpha$-specialization matrices is consistent with the $\alpha$-conjunctive rule. Likewise, for any $\alpha \in \left[ 0;1 \right] $, the distance induced by the $L_1$ matrix norm of the difference between a pair of $\alpha$-generalization matrices is consistent with the $\alpha$-disjunctive rule.

Evidential matrices are not the only representation of states of beliefs induced by $\alpha$-junctions. One can also define $\alpha$-commonality functions \cite{Sme97} $q_i^{(\alpha)}$. Let $\mathbf{Q}^{(\alpha)}$ denote the matrix obtained by $n$ Kronecker product as
\begin{equation}\label{eq:kron_alpha}
	\mathbf{Q}^{(\alpha)} = \text{Kron}\left( \begin{bmatrix} 1 & 1 \\ \alpha-1 & 1\end{bmatrix}, \hdots, \text{Kron}\left(\begin{bmatrix} 1 & 1 \\ \alpha-1 & 1\end{bmatrix}, \begin{bmatrix} 1 & 1 \\ \alpha-1 & 1\end{bmatrix} \right) \right).
\end{equation}

The vector form of function $q_i^{(\alpha)}$ is obtained as 
\begin{equation}
	\mathbf{q}_i^{(\alpha)} = \mathbf{Q}^{(\alpha)} \cdot \mathbf{m}_i,
\end{equation}
with $\mathbf{m}_i$ the vector form of some mass function $m_i$. There is a bijective correspondence between $\alpha$-commonality functions and mass functions and the $\alpha$-commonality function in correspondence with the result of an $\alpha$-conjunctive combination is equal to the entrywise product of the $\alpha$-commonality functions in correspondence with the combined mass functions. Using this property, the same reasoning as in the proof of proposition \ref{prop:q_dist_conj} applies and any $L_k$ norm based distance between $\alpha$-commonality functions is consistent with the corresponding $\alpha$-conjunctive rule. For the proof to hold, one also needs that $|q_i^{(\alpha)}\left( B \right) | \leq 1$ for any $B \subseteq \Omega$. Looking at equation (\ref{eq:kron_alpha}), we actually have
\begin{equation}
	\alpha - 1 \leq q_i^{(\alpha)}\left( B \right) \leq 1, \forall B \subseteq \Omega.
\end{equation}
Similarly, Smets \cite{Sme97} also introduces $\alpha$-implicability functions $b_i^{(\alpha)}$. Let $\mathbf{B}^{(\alpha)}$ denote the matrix obtained by $n$ Kronecker product as
\begin{equation}\label{eq:kron_alpha2}
	\mathbf{B}^{(\alpha)} = \text{Kron}\left( \begin{bmatrix} 1 & 1 \\ 1 & \alpha - 1\end{bmatrix}, \hdots, \text{Kron}\left( \begin{bmatrix} 1 & 1 \\ 1 & \alpha - 1\end{bmatrix}, \begin{bmatrix} 1 & 1 \\ 1 & \alpha - 1\end{bmatrix} \right)\right).
\end{equation}

The vector form of function $b_i^{(\alpha)}$ is obtained as 
\begin{equation}
	\mathbf{b}_i^{(\alpha)} = \mathbf{B}^{(\alpha)} \cdot \mathbf{m}_i,
\end{equation}
with $\mathbf{m}_i$ the vector form of some mass function $m_i$. There is a bijective correspondence between $\alpha$-implicability functions and mass functions and the $\alpha$-implicability function in correspondence with the result of an $\alpha$-disjunctive combination is equal to the entrywise product of the $\alpha$-implicability functions in correspondence with the combined mass functions.
Using this property, the same reasoning as in the proof of proposition \ref{prop:pl_dist_disj} applies and any $L_k$ norm based distance between $\alpha$-implicability functions is consistent with the corresponding $\alpha$-disjunctive rule. For the proof to hold, one also needs that $|b_i^{(\alpha)}\left( B \right) | \leq 1$ for any $B \subseteq \Omega$. Again, from equation (\ref{eq:kron_alpha2}), we actually have
\begin{equation}
	\alpha - 1 \leq b_i^{(\alpha)}\left( B \right) \leq 1, \forall B \subseteq \Omega.
\end{equation}
Computational difficulties are found when $n$ increases for computing $\alpha$-commonality and $\alpha$-implicability functions but this is also true for computing $\alpha$-specialization or $\alpha$-generalization matrices, therefore all the distances evoked in this appendix section are on an equal footing from both theoretical and practical considerations.


\bibliographystyle{abbrv}
\bibliography{biblio_JK}

\begin{thebibliography}{10}

\bibitem{burger2016geometric}
T.~Burger.
\newblock Geometric views on conflicting mass functions: From distances to
  angles.
\newblock {\em International Journal of Approximate Reasoning}, 70:36--50,
  2016.

\bibitem{bur13}
T.~Burger, S.~Destercke, et~al.
\newblock How to randomly generate mass functions.
\newblock {\em International Journal of Uncertainty, Fuzziness and
  Knowledge-Based Systems}, pages 645--673, 2013.

\bibitem{COUSO2014}
I.~Couso and D.~Dubois.
\newblock Statistical reasoning with set-valued information: Ontic vs.
  epistemic views.
\newblock {\em International Journal of Approximate Reasoning}, 55(7):1502 --
  1518, 2014.
\newblock Special issue: Harnessing the information contained in low-quality
  data sources.

\bibitem{couso2014random}
I.~Couso, D.~Dubois, and L.~S{\'a}nchez.
\newblock Random sets and random fuzzy sets as ill-perceived random variables.
\newblock In {\em SpringerBriefs in Computational Intelligence}. Springer,
  2014.

\bibitem{Cuz08}
F.~Cuzzolin.
\newblock A geometric approach to the theory of evidence.
\newblock {\em Systems, Man, and Cybernetics, Part C: Applications and Reviews,
  IEEE Transactions on}, 38(4):522--534, July 2008.

\bibitem{Cuz10}
F.~Cuzzolin.
\newblock Geometric conditioning of belief functions.
\newblock In {\em Proceedings of BELIEF 2010, International workshop on the
  theory of belief functions}, pages 1--6, Brest, France, 2010.

\bibitem{Cuz14b}
F.~Cuzzolin.
\newblock {\em Visions of a generalized probability theory}.
\newblock Lambert Academic Publishing, London, UK, 2014.

\bibitem{Dem67}
A.~P. Dempster.
\newblock Upper and lower probabilities induced by a multiple valued mapping.
\newblock {\em Annals of Mathematical Satistics}, 38(2):325--339, 1967.

\bibitem{Den07}
T.~Denœux.
\newblock Conjunctive and disjunctive combination of belief functions induced
  by non-distinct bodies of evidence.
\newblock {\em Artificial Intelligence}, 172(2-3):234--264, february 2008.

\bibitem{Den14}
T.~Denœux.
\newblock Likelihood-based belief function: Justification and some extensions
  to low-quality data.
\newblock {\em International Journal of Approximate Reasoning}, 55(7):1535 --
  1547, 2014.
\newblock Special issue: Harnessing the information contained in low-quality
  data sources.

\bibitem{des2013}
S.~Destercke and T.~Burger.
\newblock Toward an axiomatic definition of conflict between belief functions.
\newblock {\em IEEE transactions on cybernetics}, 43(2):585--596, 2013.

\bibitem{Dia06}
J.~Diaz, M.~Rifqi, and B.~Bouchon-Meunier.
\newblock A similarity measure between basic belief assignments.
\newblock In {\em Int. Conf. on Information Fusion (FUSION'06)}, pages 1--6,
  Florence (Italy), july 2006.

\bibitem{dubois2016basic}
D.~Dubois, W.~Liu, J.~Ma, and H.~Prade.
\newblock The basic principles of uncertain information fusion. an organised
  review of merging rules in different representation frameworks.
\newblock {\em Information Fusion}, 32:12--39, 2016.

\bibitem{Dub86b}
D.~Dubois and H.~Prade.
\newblock A set-theoretic view of belief functions: logical operations and
  approximatons by fuzzy sets.
\newblock {\em Int. Journal of General Systems}, 12(3):193--226, 1986.

\bibitem{Fre1906}
M.~Fr\'{e}chet.
\newblock {\em Sur quelques points du calcul fonctionnel}.
\newblock PhD thesis, Facult\'{e} des Sciences de Paris, 1906.

\bibitem{jac1901}
P.~Jaccard.
\newblock Distribution de la flore alpine dans le bassin des dranses et dans
  quelques r\'egions voisines.
\newblock {\em Bulletin de la Soci\'et\'e Vaudoise des Sciences Naturelles},
  37:241--272, 1901.

\bibitem{Jou01}
A.-L. Jousselme, D.~Grenier, and E.~Boss\'{e}.
\newblock A new distance between two bodies of evidence.
\newblock {\em Information Fusion}, 2:91--101, 2001.

\bibitem{Jou12}
A.-L. Jousselme and P.~Maupin.
\newblock Distances in evidence theory: Comprehensive survey and
  generalizations.
\newblock {\em International Journal of Approximate Reasoning}, 53(2):118 --
  145, 2012.
\newblock Theory of Belief Functions (BELIEF 2010).

\bibitem{kennes2013computational}
R.~Kennes and P.~Smets.
\newblock Computational aspects of the mobius transform.
\newblock {\em arXiv preprint arXiv:1304.1122}, 2013.

\bibitem{KLEIN201615}
J.~Klein, S.~Destercke, and O.~Colot.
\newblock Interpreting evidential distances by connecting them to partial
  orders: Application to belief function approximation.
\newblock {\em International Journal of Approximate Reasoning}, 71:15 -- 33,
  2016.

\bibitem{Kle14}
J.~Klein, M.~Loudahi, J.-M. Vannobel, and O.~Colot.
\newblock $\alpha$-junctions of categorical mass functions.
\newblock In F.~Cuzzolin, editor, {\em Belief Functions: Theory and
  Applications}, volume 8764 of {\em Lecture Notes in Computer Science}, pages
  1--10. Springer International Publishing, 2014.

\bibitem{lesot2008similarity}
M.-J. Lesot, M.~Rifqi, and H.~Benhadda.
\newblock Similarity measures for binary and numerical data: a survey.
\newblock {\em International Journal of Knowledge Engineering and Soft Data
  Paradigms}, 1(1):63--84, 2008.

\bibitem{LIU2006}
W.~Liu.
\newblock Analyzing the degree of conflict among belief functions.
\newblock {\em Artificial Intelligence}, 170(11):909 -- 924, 2006.

\bibitem{Lou14b}
M.~Loudahi, J.~Klein, J.-M. Vannobel, and O.~Colot.
\newblock Fast computation of $l_p$ norm-based specialization distances between
  bodies of evidence.
\newblock In F.~Cuzzolin, editor, {\em Belief Functions: Theory and
  Applications}, volume 8764 of {\em Lecture Notes in Computer Science}, pages
  422--431. Springer International Publishing, 2014.

\bibitem{Loudahi2014}
M.~Loudahi, J.~Klein, J.-M. Vannobel, and O.~Colot.
\newblock New distances between bodies of evidence based on dempsterian
  specialization matrices and their consistency with the conjunctive
  combination rule.
\newblock {\em International Journal of Approximate Reasoning}, 55(5):1093 --
  1112, 2014.

\bibitem{Lou16}
M.~Loudahi, J.~Klein, J.-M. Vannobel, and O.~Colot.
\newblock Evidential matrix metrics as distances between meta-data dependent
  bodies of evidence.
\newblock {\em Cybernetics, IEEE Transactions on}, 46(1):109--122, Jan 2016.

\bibitem{Mar08}
A.~Martin, A.-L. Jouselme, and C.~Osswald.
\newblock Conflict measure for the discounting operation on belief functions.
\newblock In {\em Proceedings of the Eleventh International Conference on
  Information Fusion}, pages 1--8, 2008.

\bibitem{matheron1975random}
G.~Math\'{e}ron.
\newblock {\em Random sets and integral geometry}.
\newblock Wiley New York, 1975.

\bibitem{miranda2005random}
E.~Miranda, I.~Couso, and P.~Gil.
\newblock Random sets as imprecise random variables.
\newblock {\em Journal of Mathematical Analysis and Applications},
  307(1):32--47, 2005.

\bibitem{molchanov2005theory}
I.~S. Molchanov.
\newblock {\em Theory of random sets}, volume~19.
\newblock Springer, 2005.

\bibitem{nguyenintroduction}
H.~T. Nguyen.
\newblock An introduction to random sets. 2006.
\newblock {\em Chapman\&Hall/CRC, Boca Raton, FL}.

\bibitem{Ngu78}
H.~T. Nguyen.
\newblock On random sets and belief functions.
\newblock {\em Journal of Mathematical Analysis and Applications}, 65(3):531 --
  542, 1978.

\bibitem{Pic12-2}
F.~Pichon.
\newblock On the $\alpha$-conjunctions for combining belief functions.
\newblock In T.~Denoeux and M.-H. Masson, editors, {\em Belief Functions:
  Theory and Applications}, volume 164 of {\em Advances in Intelligent and Soft
  Computing}, pages 285--292. Springer Berlin Heidelberg, 2012.

\bibitem{Pic09}
F.~Pichon and T.~Denoeux.
\newblock Interpretation and computation of $\alpha$-junctions for combining
  belief functions.
\newblock In {\em 6th Int. Symposium on Imprecise Probability: Theories and
  Applications (ISIPTA '09)}, Durham, U.K., 2009.

\bibitem{pichon2019several}
F.~Pichon, A.-L. Jousselme, and N.~B. Abdallah.
\newblock Several shades of conflict.
\newblock {\em Fuzzy Sets and Systems}, 366:63 -- 84, 2019.

\bibitem{schmelzer2012characterizing}
B.~Schmelzer.
\newblock Characterizing joint distributions of random sets by multivariate
  capacities.
\newblock {\em International Journal of Approximate Reasoning},
  53(8):1228--1247, 2012.

\bibitem{schmelzer2015joint}
B.~Schmelzer.
\newblock Joint distributions of random sets and their relation to copulas.
\newblock {\em International Journal of Approximate Reasoning}, 65:59--69,
  2015.

\bibitem{schmelzer2015sklar}
B.~Schmelzer.
\newblock Sklar's theorem for minitive belief functions.
\newblock {\em International Journal of Approximate Reasoning}, 63:48--61,
  2015.

\bibitem{schmelzer2018multivariate}
B.~Schmelzer.
\newblock Multivariate capacity functionals vs. capacity functionals on product
  spaces.
\newblock {\em Fuzzy Sets and Systems}, 364:1--35, 2019.

\bibitem{Sha76}
G.~Shafer.
\newblock {\em A Mathematical Theory of Evidence}.
\newblock Princeton University press, Princeton (NJ), USA, 1976.

\bibitem{Sme97}
P.~Smets.
\newblock The alpha-junctions: Combination operators applicable to belief
  functions.
\newblock In {\em First Int. Joint Conference on Qualitative and Quantitative
  Practical Reasoning (ECSQUARU-FAPR'97)}, volume 1244 of {\em Lecture Notes in
  Computer Sciences}, pages 131--153. Springer International Publishing, Bad
  Honef, Germany, 1997.

\bibitem{Sme02}
P.~Smets.
\newblock The application of the matrix calculus to belief functions.
\newblock {\em International Journal of Approximate Reasoning}, 31(1–2):1 --
  30, 2002.

\bibitem{Sme05}
P.~Smets.
\newblock Belief functions on real numbers.
\newblock {\em International Journal of Approximate Reasoning}, 40:181--223,
  2005.

\bibitem{Sme94}
P.~Smets and R.~Kennes.
\newblock The transferable belief model.
\newblock {\em Artificial Intelligence}, 66(2):191--234, 1994.

\bibitem{TESSEM1993}
B.~Tessem.
\newblock Approximations for efficient computation in the theory of evidence.
\newblock {\em Artificial Intelligence}, 61(2):315 -- 329, 1993.

\bibitem{yager1986entailment}
R.~R. Yager.
\newblock The entailment principle for dempster—shafer granules.
\newblock {\em International Journal of Intelligent Systems}, 1(4):247--262,
  1986.

\bibitem{Z-D}
L.~M. Zouhal and T.~Den{\oe}ux.
\newblock An evidence-theoretic k-nn rule with parameter optimisation.
\newblock {\em IEEE Transactions on Systems, Man and Cybernetics. Part C:
  Application and reviews}, 28(2):263 -- 271, 1998.

\end{thebibliography}

\end{document}